\documentclass[a4paper,11pt,onecolumn]
               {quantumarticle}
               
\pdfoutput=1

\usepackage[english]{babel}
\usepackage{amsmath}
\usepackage{amssymb}  
\usepackage{amsthm}
\usepackage{amsfonts}
\usepackage{graphicx}
\usepackage{bbm}
\usepackage{url}
\usepackage{tikz}
\usepackage{comment}
\definecolor{Red}{HTML}{E53E30}  %
\definecolor{Green}{HTML}{00AD69}  %
\definecolor{Blue}{HTML}{2171b5}
\definecolor{Purple}{HTML}{652F6C}  %

\usepackage{braket}
\usepackage{ upgreek }
\usepackage{dsfont}
\usepackage{algpseudocode}
\usepackage{algorithm}

\usepackage{hyperref}

\theoremstyle{plain}               
\newtheorem{thm}{Theorem}[section]
\newtheorem{fact}{Fact}[section]
\newtheorem{lem}{Lemma}[section]
\newtheorem{cor}{Corollary}[section]
\newtheorem{prop}{Proposition}[section]

\newtheorem{defi}{Definition}[section]

\definecolor{Green}{HTML}{00AD69}  %

\def\beq{\begin{equation}}
\def\eeq{\end{equation}}
\def\bq{\begin{quote}}
\def\eq{\end{quote}}
\def\ben{\begin{enumerate}}
\def\een{\end{enumerate}}
\def\bit{\begin{itemize}}
\def\eit{\end{itemize}}

\def\lb{\left(}
\def\rb{\right)}

\def\l|{\left|}
\def\r|{\right|}

\newcommand\C{\mathbbm{C}}

\newcommand\M{\mathcal{M}}

\newcommand{\ketbra}[1]{|#1\rangle\langle#1|}
\newcommand{\tr}[1]{\text{tr}\lb#1\rb}

\newcommand{\rl}[2]{S\lb#1\|#2\rb}

\newcommand{\setC}{\mathbb{C}}

\newcommand{\cO}{\mathcal{O}}
\newcommand{\tcO}{\tilde{\mathcal{O}}}

\begin{document}

\title{Fast and robust quantum state tomography from few basis measurements}

\author{Daniel Stilck Fran\c{c}a}
\affiliation{QMATH, Department of Mathematical Sciences, University of Copenhagen, DK 2100}

\author{Fernando G.S L. Brand\~{a}o}
\affiliation{AWS Center for Quantum Computing, Pasadena, CA 91125}
\affiliation{Institute for Quantum Information and Matter, California Institute of Technology, Pasadena,~CA~91125}

\author{Richard Kueng}
\affiliation{Institute for Integrated Circuits, Johannes Kepler University Linz, AT 4040}

\begin{abstract}
Quantum state tomography is a powerful, but resource-intensive, general solution for numerous quantum information processing tasks. This motivates the design of robust tomography procedures that use relevant resources as sparingly as possible. Important cost factors include the number of state copies and measurement settings, as well as classical postprocessing time and memory. In this work, we present and analyze an online tomography algorithm designed to optimize all the aforementioned resources at the cost of a worse dependence on accuracy.
The protocol is the first to give provably optimal performance in terms of rank and dimension for state copies, measurement settings and memory.
Classical runtime is also reduced substantially and numerical experiments demonstrate a favorable comparison with other state-of-the-art techniques.
Further improvements are possible by executing the algorithm on a quantum computer, giving a quantum speedup for quantum state tomography.
\end{abstract}

\maketitle
\flushbottom

\section{Motivation}

Quantum state tomography is the task of reconstructing a classical description of a quantum state from experimental data.
This problem has a long and rich history \cite{Gross2013} and remains a useful subroutine for building, calibrating and controlling quantum information processing devices. 
Over the last decade, unprecedented advances in the experimental control of quantum architectures have pushed traditional estimation techniques 
to the limit of their capabilities. This is mainly due to a fundamental curse of dimension: the dimension of state space grows exponentially in the number of qudits, i.e.\ a quantum system comprised of $n$ $d$-dimensional qudits is characterized by a density matrix $\rho$ of size $D=d^n$. 
The impact of this scaling behavior is further amplified by the probabilistic nature of quantum mechanics (``wave-function collapse''). Information about the state is only accessible via measuring the system. An informative quantum measurement is destructive and only yields probabilistic outcomes. Hence, many identically prepared samples of the quantum state are required to estimate even a single parameter of the underlying state. Characterizing the full state of a quantum system necessitates accurate estimation of many such parameters. 
Storing and processing the measurement data also requires
substantial amounts of classical memory and computing power -- another important practical bottleneck.
To summarize: the curse of dimension and wave-function collapse have severe implications that necessitate the design of extremely resource-efficient protocols.

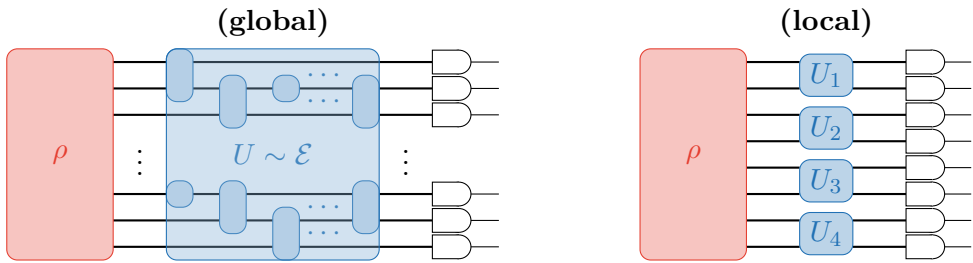
\begin{figure}
    \centering
\begin{tabular}{lcr}
\begin{tikzpicture}[baseline,scale=0.7]
\foreach \y in {-1.75,-1.25,-0.75,0.75,1.25,1.75}
{
\draw[thick] (-3,\y) -- (3,\y);
\draw (3.5,\y) circle (0.225);
\fill[white] (3,\y-0.225) rectangle (3.5,\y+0.225);
\draw (3.5,\y-0.225) -- (3,\y-0.225) -- (3,\y+0.225) -- (3.5,\y+0.225);
\draw (3.725,\y) -- (4.25,\y);
}
\node at (-2.5,0) {$\vdots$};
\node at (2.5,0) {$\vdots$};
\node at (0,0) {\textcolor{Blue}{$U \sim \mathcal{E}$}};
\draw[rounded corners,Red] (-5,-2) rectangle (-3,2);
\fill[rounded corners,Red,opacity=0.3]
(-5,-2) rectangle (-3,2);
\node at (-4,0) {\textcolor{Red}{$\rho$}};
\fill[white, rounded corners] (-2,1.0) rectangle (-1.5,2);
\draw[Blue, rounded corners] (-2,1) rectangle (-1.5,2);
\fill[Blue, rounded corners,opacity=0.2] (-2,1) rectangle (-1.5,2);
\fill[white, rounded corners] (-1,0.5) rectangle (-0.5,1.5);
\fill[Blue, rounded corners,opacity=0.2] (-1,0.5) rectangle (-0.5,1.5);
\draw[Blue, rounded corners] (-1,0.5) rectangle (-0.5,1.5);
\fill[white, rounded corners] (0,1) rectangle (0.5,1.5);
\fill[Blue, rounded corners,opacity=0.2] (0,1) rectangle (0.5,1.5);
\draw[Blue, rounded corners] (0,1) rectangle (0.5,1.5);
\fill[white, rounded corners] (1.5,0.5) rectangle (2,1.5);
\fill[Blue, rounded corners,opacity=0.2] (1.5,0.5) rectangle (2,1.5);
\draw[Blue, rounded corners] (1.5,0.5) rectangle (2,1.5);
\node at (1,1) {\textcolor{Blue}{$\cdots$}};
\node at (1,1.5) {\textcolor{Blue}{$\cdots$}};
\fill[white,rounded corners] (-2,-1) rectangle (-1.5,-0.5);
\fill[Blue,rounded corners,opacity=0.2] (-2,-1) rectangle (-1.5,-0.5);
\draw[Blue,rounded corners] (-2,-1) rectangle (-1.5,-0.5);
\fill[white,rounded corners] (-1,-1.5) rectangle (-0.5,-0.5);
\fill[Blue,rounded corners,opacity=0.2] (-1,-1.5) rectangle (-0.5,-0.5);
\draw[Blue,rounded corners] (-1,-1.5) rectangle (-0.5,-0.5);
\fill[white,rounded corners] (0,-2) rectangle (0.5,-1);
\fill[Blue,rounded corners,opacity=0.2] (0,-2) rectangle (0.5,-1);
\draw[Blue,rounded corners] (0,-2) rectangle (0.5,-1);
\fill[white,rounded corners] (1.5,-1.5) rectangle (2,-0.5);
\fill[Blue,rounded corners,opacity=0.2] (1.5,-1.5) rectangle (2,-0.5);
\draw[Blue,rounded corners] (1.5,-1.5) rectangle (2,-0.5);
\node at (1,-1) {\textcolor{Blue}{$\cdots$}};
\node at (1,-1.5) {\textcolor{Blue}{$\cdots$}};
\fill[white,rounded corners,opacity=0.2] (-2,-2) rectangle (2,2);
\fill[Blue,rounded corners,opacity=0.2] (-2,-2) rectangle (2,2);
\draw[Blue,rounded corners] (-2,-2) rectangle (2,2);
\node at (0,2.5) {\textbf{(global)}};
\end{tikzpicture}
& \hspace{1cm} &
\begin{tikzpicture}[baseline,scale=0.7]
\foreach \y in {-1.75,-1.25,-0.75,-0.25,0.25,0.75,1.25,1.75}
{
\draw[thick] (-1.5,\y) -- (-0.5,\y);
\draw[thick] (0.5,\y)-- (1.5,\y);
\draw (2,\y) circle (0.225);
\fill[white] (1.5,\y-0.225) rectangle (2,\y+0.225);
\draw (2,\y-0.225) -- (1.5,\y-0.225) -- (1.5,\y+0.225) -- (2,\y+0.225);
\draw (2.225,\y) -- (2.725,\y);
}
\foreach \y in {-2,-1,0,1}
{
\draw[Blue, rounded corners] (-0.5,\y+0.1) rectangle (0.5,\y+0.9);
\fill[Blue, rounded corners,opacity=0.3] (-0.5,\y+0.1) rectangle (0.5,\y+0.9);
}
\draw[rounded corners,Red] (-3.5,-2) rectangle (-1.5,2);
\fill[rounded corners,Red,opacity=0.3]
(-3.5,-2) rectangle (-1.5,2);
\node at (-2.5,0) {\textcolor{Red}{$\rho$}};
\node at (0,1.5) {\textcolor{Blue}{$U_1$}};
\node at (0,0.5) {\textcolor{Blue}{$U_2$}};
\node at (0,-0.5) {\textcolor{Blue}{$U_3$}};
\node at (0,-1.5) {\textcolor{Blue}{$U_4$}};
\node at (0,2.5) {\textbf{(local)}};
\end{tikzpicture}
\end{tabular}
    \caption{\emph{Basis measurement primitive.} Global measurements (right) require implementing a global unitary that affects all qubits prior to measuring in the computational basis. A $k$-local measurement primitive only allows for unitaries that affect groups of $k$ (geometrically) local qubits; see the left-hand side for a visualization with $k=2$.}
    \label{fig:illustration}
\end{figure}

In this work, we focus on 
reconstructing the complete density matrix $\rho$ from single-copy measurements.
This is an actual restriction, as it excludes some of the most powerful tomography techniques known to this date \cite{Wright2016, Haah2017}. While very efficient in terms of state copies, these procedures are very demanding in terms of quantum hardware -- an actual implementation would require exponentially long quantum circuits that act
collectively on all the copies of the unknown state stored in a quantum memory.

We also adopt a measurement primitive that mimics the layout of modern quantum information processing devices. 
Apply a unitary $U$ to the unknown state $\rho \mapsto U \rho U^\dagger$ and perform measurements in the computational basis $\left\{|i \rangle:\;i=1,\ldots,D\right\}$. Fixing $U$ and repeating this procedure many times allows for estimating the associated outcome distribution:
\begin{equation}
\left[p_U (\rho) \right]_i = \langle i| U \rho U^\dagger |i \rangle \quad \text{for $i=1,\ldots,D$.} \label{eq:diagonal-measurement}
\end{equation}
This outcome distribution characterizes the diagonal elements of $U \rho U^\dagger$. In general, access to a single diagonal is insufficient to determine $\rho$ unambiguously. Instead, multiple repetitions of this basic measurement primitive are necessary. 
We refer to Fig.~\ref{fig:illustration} for an illustration.
Different ensembles $\mathcal{E}$ of accessible unitary transformations give rise to different basis measurement primitives.
When employed to perform state tomography -- i.e.\ reconstruct an unknown state $\rho$ up to accuracy $\epsilon$ in trace distance --
the following fundamental scaling laws apply to \emph{any} (single-copy) basis measurement primitive and \emph{any} tomographic procedure:
\begin{enumerate}
    \item[i.] The \emph{number of basis measurement settings} $M$ must scale at least linearly with the (effective) target rank $r = \mathrm{rank}(\rho)$: $M = \Omega (r)$. This corresponds to estimating a total of $DM = \Omega (rD)$ parameters \cite{Heinosaari2013,Kech2017}.
    \item[ii.] The \emph{sampling rate} $N$, i.e.\ the number of independent state copies required to obtain sufficient data, must depend on rank, dimension and desired accuracy: $N = \Omega \left( D r^2 /\epsilon^2\right)$ \cite{Haah2017}.
    \item[iii.] The \emph{classical storage} $S$ is bounded by dimension times target rank: $S = \Omega (rD)$.
\end{enumerate}
Constraint iii.\ follows from a simple parameter counting argument -- specifying a general $D \times D$-matrix with rank $r$ requires (order) $rD$ parameters -- while i.\ and ii.\ reflect fundamental limitations that
have only been identified comparatively recently. 
These bounds cover three of the four most relevant cost parameters. For the last one we are not aware of a nontrivial rigorous lower bound:
\begin{enumerate}
\item[iv.] The \emph{classical runtime} associated with processing the measurement data to produce an estimated state $\sigma_\star$ should be as fast as possible.
\end{enumerate}

The last decade has seen the development of several procedures that provably optimize (at least) some of these four cost factors up to logarithmic factors in the ambient dimension. We refer to Table~\ref{tab:comparison} for a detailed tabulation of resource requirements. For now, we content ourselves with emphasizing that existing procedures have been designed to either minimize the number of measurement settings (compressed sensing approaches~\cite{Gross2010,liu2011universal,kueng_low_2017}) or the required number of samples per measurement (least-squares approaches~\cite{Sugiyama2013,guta_fast_2018}). Neither of these approaches seems to be well-suited for optimizing classical postprocessing memory and time.
Finally, we point out that currently available quantum technologies are not perfect \cite{Preskill2018}. Practical tomography procedures should be \emph{robust} with respect to imperfections, most notably state preparation and measurement errors.

\section{Overview of results}\label{sec:overview}

\begin{table*}
\begin{center}
{\footnotesize
\begin{tabular}{|c|c|c|c|c|c|}
\hline
  & \textbf{meas.\ primitive}  & \textbf{basis settings} & \textbf{state copies}  & \textbf{runtime} & \textbf{memory} \\
\hline
\hline
lower bounds & arbitrary & $\geq r$ & $ \geq Dr^2 \epsilon^{-2}$ & $ \geq Dr^2 \epsilon^{-2}$ & $ \geq Dr$ \\
 \hline
 CS \cite{Voroninski2013}
 & Haar & $r$ & unknown & $D^4$ & $D^3$  \\
 \hline
 CS 
 \cite{kueng2015} & Clifford & $D^{2/3} r$ & unknown & 
 $D^4$ & $D^3$ \\
 \hline
 PLS \cite{guta_fast_2018} & 2-design & $D$ & $D r^2 \epsilon^{-2}$ & $D^3$ & $D^2$ \\
 \hline
 this work & 4-design & $r \epsilon^{-2}$ & $ D r^2 \epsilon^{-4}$ & $D^2 r^{5/2} \epsilon^{-5}$ & $D r \epsilon^{-2}$\\
 \hline
 this work & Clifford & $ r^3 \epsilon^{-2}$ & $D r^4 \epsilon^{-4}$ & $D^2r^6 \epsilon^{-5}$ & $Dr^2 \epsilon^{-2}$\\
 \hline
\end{tabular}
}
\end{center}
\caption{
\emph{Resource scaling for state tomography protocols based on global measurements (single copy):} Here, $D$ denotes the Hilbert space dimension, $r$ is the rank of the target state and $\epsilon$ is the desired precision (in trace distance). We have suppressed constants, as well as logarithmic dependencies in $D$ and $r$.
The first row summarizes known fundamental lower bounds, while the label ``unknown'' indicates a lack of rigorous theory support.
}
\label{tab:comparison}
\end{table*}

In this work, we develop a robust algorithm for almost resource-optimal quantum state tomography from (single-copy) basis measurements
that comes with rigorous convergence guarantees.
The theoretical results are closely related to quantum state distinguishability \cite{Holevo1973,Helstrom1969,ambainis_wise_2007,Matthews_2009} and strongest for global measurement primitives (Fig.~\ref{fig:illustration}, left) that are sufficiently generic.
In the regime of low target rank $r$, the proposed method improves upon state-of-the art techniques at the cost of a worse dependence on target accuracy $\epsilon$. 
The actual numbers are summarized in Table~\ref{tab:comparison}. 
The required number of basis measurement setting matches results from compressed sensing~\cite{Gross2010,liu2011universal,kueng_low_2017} -- a technique that has been specifically designed to optimize this cost function -- while the required number of state copies is comparable to projected least squares~\cite{Sugiyama2013,guta_fast_2018} -- which is known to be (almost) optimal in this regard. Classical runtime and memory cost are also reduced substantially.
We also obtain rigorous results for $k$-local measurement primitives (Fig.~\ref{fig:illustration}, right), but the obtained theoretical numbers only become competitive if the locality parameter $k$ is sufficiently large.
We believe that this shortcoming is an artifact of poor constants and refer to App.~\ref{sub:local-measurements} for details. 

\subsection{Algorithm and theoretical runtime guarantee}

The tomography algorithm -- which we call \emph{Hamiltonian updates} -- is based on a variant of the versatile mirror-descent meta-algorithm \cite{Tsuda2005,Bubeck2015}, see also \cite{Brandao2019sdp}. Mirror descent and its cousin, matrix multiplicative weights, have led to considerable progress in algorithm design across several disciplines. Prominent examples include fast semidefinite programming solvers ~\cite{Hazan2006,Arora2007,Steurer2015,Apeldoorn2017,Brandao2017a,Brandao2017b,Brandao2019sdp}, quantum prediction techniques like shadow tomography~\cite{Aaronson_2018}, the online learning methods of~\cite{Aaronson2019} and the tomography protocol of~\cite{Youssry_2019}.
The algorithm design is summarized in Algorithm~\ref{alg:motivation}.
The key idea is to maintain and iteratively update a guess for the unknown state.
The sequence of guess states is parametrized by Hamiltonians
\begin{align*}
\sigma_t =& \frac{\exp (-H_t)}{\mathrm{tr} (\exp (-H_t))} \quad \text{for} \quad  t=0,1,2,\ldots 
& \text{(Gibbs / thermal state)}
\end{align*}
and initialized to an infinite temperature state $\sigma_0 = \mathbb{I}/D$ (maximum entropy principle). At each subsequent iteration, we choose a unitary rotation $U\sim \mathcal{E}$ \emph{at random} from a fixed ensemble, estimate the outcome distribution \eqref{eq:diagonal-measurement} of the rotated target state $U\rho U^\dagger$ and compare it to the predicted outcome distribution of the current guess $\sigma_t$. If the two outcome distributions differ by more than mere statistical fluctuations, $\sigma_t$ is an inadequate guess for $\rho$.

\begin{algorithm}[t]
\label{alg:HUtomo}
\caption{\textit{Hamiltonian Updates for quantum 
state tomography }
}
\begin{algorithmic}[H]
\State \textbf{Input:} error tolerance $\epsilon$,  number of loops $L$.
\State \textbf{Initialize:} $t=0$, $H_t=0$, \textsc{convergence}=\textsc{false}
\While{\textsc{convergence}=\textsc{false}}
\State compute $\sigma_t = \exp(-H_t)/\mathrm{tr}(\exp(-H_t))$  \Comment current guess for the state $\rho$
\State select random basis measurement
$\left\{ U|i \rangle \! \langle i|U^\dagger \right\}$
\State compute 
outcome statistics $[p_i]$ of $\sigma_t$ 
 \Comment classical computation
\State estimate  
outcome statistics $[q_i]$ of $\rho$  \Comment quantum measurement
\State \textbf{check} if $[p_i]$ and $[q_i]$ are $\epsilon$-close in $\ell_1$ distance
\If{\textsc{no}}{ set
$P=\sum_{p_i>q_i} |i \rangle \! \langle i|$} 
\Comment{collect outcomes for which $p_i>q_i$}
\State Set $\eta=\frac{1}{8}\|p-q\|_{\ell_1}$
\State $H_{t+1} \gets H_t + \eta U^\dagger P U$ \Comment energy penalty for mismatch (in this basis)
\State update $\sigma_{t+1} = \exp(-H_{t+1})/\mathrm{tr}(\exp(-H_{t+1}))$ 
\State  $t\gets t+1$ \Comment{update counter of number of iterations}

\ElsIf{\textsc{yes}}  \Comment current guess may be close to $\rho$
\State check $L$ additional random bases \Comment{suppress likelihood of false positives}
\If {$\ell_1$ distance is always $<\epsilon$ } \Comment current guess is likely to be close
\State{\textbf{set} \textsc{convergence=true}}
\EndIf
\EndIf

\EndWhile

\State \textbf{Output:} $H_t$
\end{algorithmic}
\label{alg:motivation}
\end{algorithm}

We then update the guess state $\sigma_t \mapsto \sigma_{t+1}$ by including a small energy penalty in the associated Hamiltonian that penalizes the observed mismatch and repeat. 
Heuristically, it is reasonable to expect that this update rule makes progress as long as each newly selected basis provides actionable advice, i.e.\ discrepancies in the outcome distributions. Things get more interesting when this is not the case. 
Predicted and estimated outcome distribution can be very close for two reasons (i): the current iterate $\sigma_t$ is close to the unknown target
(\emph{convergence}); (ii.) the current basis measurement cannot properly distinguish between $\sigma_t$ and $\rho$, even though they are still far apart (\emph{false positive}). It is imperative to protect against wrongfully terminating the procedure due to the occurrence of a false positive. Hamiltonian Updates (Algorithm~\ref{alg:motivation}) 
suppresses the likelihood of wrongfully terminating by checking closeness in (up to) $L$ additional random bases.
The required size of such a control loop depends on the measurement primitive. 
Broadly speaking, generic measurement ensembles -- like Haar-random unitary transformations -- are very unlikely to produce false positives; while highly structured ensembles -- 
like mutually unbiased bases
-- can be much more susceptible.
The following relation introduces two ensemble-dependent summary parameters that capture this effect:
\begin{equation}
 \mathrm{Pr}_{U \sim \mathcal{E}} \left[ \| p_U (\rho) - p_U (\sigma_t) \|_{\ell_1} \geq \theta_{\mathcal{E}}(\rho,\sigma_t)\|\rho-\sigma_t\|_2\right] 
\geq  \tau_{\mathcal{E}}(\rho,\sigma_t). \label{eq:tomo-sound-intro}
\end{equation}
The parameter $\theta_{\mathcal{E}}(\rho,\sigma_t)$ relates an observed discrepancy in outcome distributions (measured in $\ell_1$ distance) to the Frobenius distance in state space. As detailed below, it captures the minimal progress we can expect from a successful update $\sigma_t \mapsto \sigma_{t+1}$. The second parameter $\tau_{\mathcal{E}}(\rho,\sigma_t)$ lower bounds the probability of observing an outcome discrepancy that appropriately reflects the current stage of convergence. This parameter controls the size of the control loop. It is desirable to choose both parameters as large as possible, but there is a trade-off (making $\theta_{\mathcal{E}}(\rho,\sigma_t)$ larger necessarily diminishes $\tau_{\mathcal{E}}(\rho,\sigma_t)$) and both depend heavily on the measurement ensemble. 
%
One of our main theoretical contributions is a rigorous convergence guarantee for Hamiltonian updates (Algorithm~\ref{alg:motivation}) that only depends on the ambient dimension $D$, the target rank $r = \mathrm{rank}(\rho)$, as well as the worst-case ensemble parameters
\begin{align}\label{equ:measurement_ensemble_par}
\theta_{\mathcal{E}}(\rho) = \max_{\sigma \text{ state}} \theta_{\mathcal{E}}(\rho,\sigma) \quad \text{and} \quad  
\tau_{\mathcal{E}}(\rho) = \max_{\sigma \text{ state}} \tau_{\mathcal{E}}(\rho,\sigma). 
\end{align}

\begin{thm}[informal statement] \label{thm:main-informal}
Fix a measurement primitive $\mathcal{E}$, a desired accuracy $\epsilon$ and let $\rho$ be a rank-$r$ target state.
With high probability, Algorithm~\ref{alg:motivation}
requires at most $T=\mathcal{O}\left(r \log (D)/ (\theta_{\mathcal{E}}(\rho) \epsilon)^2\right)$ steps -- each with a control loop of size $L=\mathcal{O}(\log (T)/\tau_{\mathcal{E}}(\rho))$ -- to produce an output $\sigma_{\star}$ that obeys $\| \rho-\sigma_{\star} \|_1 \leq \epsilon$.
\end{thm}

This convergence guarantee is also stable with respect to imperfect implementations. 
In particular, we only need to estimate measurement outcome statistics to a certain degree of accuracy: $\mathcal{O}\left(D r /(\theta_{\mathcal{E}}(\rho) \epsilon)^2 \right)$ measurement repetitions suffice for each basis.
This implies that the total number of measurement settings and state copies are bounded by
\begin{align}
M =& TL \simeq \mathcal{O} \left( r \log (D)/(\tau_{\mathcal{E}}(\rho)\theta_{\mathcal{E}}(\rho)^2 \epsilon^2)\right)  &\text{(measurement settings)}, \label{eq:basis-complexity} \\
N \simeq& \mathcal{O} \left( D r^2 \log (D)/(\tau_{\mathcal{E}}(\rho) \theta_{\mathcal{E}} (\rho)^4 \epsilon^4) \right) & \text{(sample complexity)}. 
\label{eq:sample-complexity}
\end{align}
To increase readability, we have suppressed the logarithmic contribution in $T$.

\subsection{Connections to quantum state distinguishability}

The bounds for $M$ in Eq.~\eqref{eq:basis-complexity} and $N$ in Eq.~\eqref{eq:sample-complexity}
are characterized by worst-case ensemble parameters \eqref{eq:tomo-sound-intro}. These are intimately related to quantum state distinguishability: how good is a fixed measurement primitive $\mathcal{E}$ at distinguishing state $\rho$ from state $\sigma$ in the single-shot limit?
Ambainis and Emerson~\cite{ambainis_wise_2007} showed that the optimal probability of successful discrimination is given by $
p_{\mathrm{succ}}=\frac{1}{2}+\frac{1}{4}\mathbb{E}_{U \sim \mathcal{E}} \|p_U (\rho) - p_U (\sigma) \|_{\ell_1}
$ and achieved by the maximum likelihood rule, see also \cite{Matthews_2009}. It is possible to relate this bias to the Frobenius distance in state space:
\begin{equation*}
\mathbb{E}_{U \sim \mathcal{E}} \|p_U (\rho) - p_U (\sigma) \|_{\ell_1} \geq \lambda_{\mathcal{E}}(\rho,\sigma) \| \rho - \sigma \|_2.
\end{equation*}
The proportionality constant $\lambda_{\mathcal{E}} (\rho,\sigma)$ measures how well the measurement primitive is equipped to distinguish $\rho$
 from $\sigma$. It is closely related to the ensemble parameters defined in Eq.~\eqref{eq:tomo-sound-intro} and has been the subject of considerable attention in the community. Tight bounds have been derived for a variety of measurement primitives, such as
Haar random unitaries and approximate 4-designs \cite{ambainis_wise_2007,Matthews_2009}, random Clifford unitaries \cite{kueng_2016_distinguishing} and $k$-local (approximate) 4-designs \cite{Lancien_2013}. 
A simple probabilistic arguments allows for converting these assertion into lower bounds on both $\theta_{\mathcal{E}}(\rho)$ and $\tau_{\mathcal{E}}(\rho)$.
Inserting these bounds into Eq.~\eqref{eq:basis-complexity} and Eq.~\eqref{eq:sample-complexity} then implies the measurement and sample complexity assertions advertised in Table~\ref{tab:comparison}. We refer to Appendix~\ref{sec:examples} for a detailed case-by-case analysis and content ourselves with with an overview.
We start with the strongest measurement primitive: Haar random unitaries and 
approximate 4-designs  achieve $\theta_{\mathcal{E}}(\rho), \tau_{\mathcal{E}}(\rho)=\mathrm{const}$ for any target state. Hence, $M=\mathcal{O}(r \log (D))/\epsilon^2)$ basis settings and $N = \mathcal{O}(D r^2 \log (D)/\epsilon^4)$ state copies suffice. Clifford random measurements achieve $\theta_{\mathcal{E}}(\rho)\sim r^{-\frac{1}{2}}, \tau_{\mathcal{E}}(\rho)\sim r^{-2}$. That is, they only have a slightly worse dependency on the rank, but perform as well as Haar measurements in terms of the ambient dimension. 
On the other hand, more local measurement settings defined by unitaries acting on at most $k$ qubits have $\theta_{\mathcal{E}}(\rho)\sim \mathrm{exp}(-\cO(n/k)), \tau_{\mathcal{E}}(\rho)\sim \mathrm{exp}(-\cO(n/k))$, showing an (exponentially) worse dependency on the number of qudits when compared to Haar measurements. 
Empirical studies below do, however, suggest a much more favorable performance in practice.

This scaling highlights both a core strength and a core weakness of Hamiltonian updates.
In terms of dimension $D$ and rank $r$, these numbers saturate  fundamental lower bounds on any tomographic procedure up to a logarithmic factor.
However, the number of measurement settings also depends inverse quadratically on the accuracy. In turn, the  accuracy enters as $\epsilon^{-4}$, not $\epsilon^{-2}$ in the sample complexity. Thus, high accuracy solutions do not only require many samples, but also many basis measurement settings.
This drawback is a consequence of a ``curse of mirror descent (or multiplicative weights)". These meta-algorithms are very efficient in terms of problem dimension, but scale comparatively poorly in accuracy \cite{Arora2007}.
However, inverse polynomial scaling in accuracy $\epsilon$ is an unavoidable feature of quantum state tomography. Hence, tomography is a reasonable setting to apply algorithms that trade dimensional dependency for accuracy. Moreover, for most applications, it suffices to recover the state up to precision $\epsilon=\cO(\text{polylog}(D)^{-1})$.

\section{Summary and comparison to relevant existing work}

We propose a variant of mirror descent~\cite{Tsuda2005,Bubeck2015} to obtain resource-efficient algorithms for quantum state tomography.
In recent years, mirror descent and its cousins have been extensively used to obtain fast SDP solvers ~\cite{Hazan2006,Arora2007,Steurer2015,Apeldoorn2017,Brandao2017a,Brandao2017b,Brandao2019sdp}, to develop prediction algorithms like shadow tomography~\cite{Aaronson_2018}, the online learning methods of~\cite{Aaronson2019} and the tomography protocol of~\cite{Youssry_2019}. 
Key advantages of such an approach are resource efficiency, as well as intrinsic resilience towards noise.
Empirical studies summarized in Fig.~\ref{fig:noise} confirm these theoretical assertions.
A downside is, however, that the number of iterations may depend on the desired target accuracy $\epsilon$.
We focus on obtaining a $\epsilon$-approximation in trace distance of 
a $D$-dimensional state $\rho$ from (random) basis measurements on i.i.d.~copies (\emph{global classical description}).
Our goal is to optimize the different resources required for that task. These include the number of state copies (sample complexity), the cost for processing measurement data (classical postprocessing), as well as the associated memory cost.
The multipronged resource efficiency of our results becomes particularly pronounced if the underlying target state has (approximately) low-rank $r \ll D$. 
This is a natural assumption in most applications, but can also be relaxed to states  with low R\'enyi entropy, see App.~\ref{sec:effective-rank}.

Thus, our results are similar in spirit to the tomography algorithms based on compressed sensing (CS)~\cite{Gross2010,liu2011universal,Flammia2012,Riofr_o_2017,kueng_low_2017}, or projected least squares (PLS)~\cite{Sugiyama2013,guta_fast_2018}. These also focus on rigorous and (nearly) optimal sample complexity in the low-rank regime combined with efficient postprocessing. 
Table~\ref{tab:comparison} summarizes the resources required for these protocols, as well as our new results. These compare favorably with existing methods. 
 We note that for approximate $4$-design measurements,
 both sample complexity and memory -- as functions of $D$ and $r$ -- are essentially optimal~\cite{Wright2016,Haah2017}. Compared to existing approaches, we obtain significant savings in both runtime and memory. Moreover, as pointed out in \cite{Youssry_2019}, there are also  qualitative advantages.

Current schemes that minimize the number of basis settings~\cite{Voroninski2013, kueng2015} are only known to do so with perfect knowledge of the underlying measurement outcomes. This will never be the case in practice, due to statistical fluctuations. Thus, to the best of our knowledge, our work is the first to rigorously obtain recovery guarantees with imperfect knowledge of outcomes and basis settings that only scale logarithmically with the ambient dimension and linearly with rank (albeit with the extra $\epsilon$ dependency). 

The focus of this work differs from other recent applications of mirror descent to quantum learning~\cite{Aaronson2019,Aaronson_2018,Brandao2017b}. Broadly speaking, these works focus on obtaining a classical description of the state -- a shadow -- that approximately reproduces a fixed set of target observables. 
This is a different and weaker form of recovery. Moreover, these works prioritize sample complexity; not necessarily classical postprocessing resources.
Minimizing these classical resources is a core focus of this work.

Having said this, the idea of using (variants of) mirror descent for quantum state (and process) tomography is not completely new. Similar ideas were proposed in Refs.\ \cite{ferrie2015self-guided,granade2017adaptive} and have been experimentally tested \cite{ferrie2016experimental,ferrie2020experimental}. 
More recently,
Youssry, Tomamichel and Ferrie 
proposed and analyzed state tomography based on matrix exponentiated gradient descent \cite{Youssry_2019}. 
They focused on the practically relevant case of local (single-qubit) Pauli measurements and established
convergence to the target state as the number of samples goes to infinity. They also pointed out conceptual advantages, such as online implementation and noise-robustness. 
The results presented here add to this promising picture. 
We equip (a variant of) mirror descent with rigorous performance guarantees in the non-asymptotic setting, optimize actual implementations and establish robustness in a more general setting. Moreover, our results apply to any measurement procedure that is capable of distinguishing arbitrary pairs of quantum states.

We also want to point out that the method presented here could also be implemented on a quantum computer. This would result in substantial runtime savings -- a quantum speedup for quantum state tomography.
Suppressing polylogarithmic terms, a runtime of order $\tcO(D^{\frac{3}{2}}r^3\epsilon^{-9})$
suffices to obtain a \emph{classical} description of the target state. We refer to App.~\ref{sec:quantumimplementation} for details and proofs. To the best of our knowledge, this is the first quantum speedup for low-rank tomography beyond the results of Kerenidis and Prakash~\cite{kerenidis2018quantum} which cover pure, real  target states ($r=1$) exclusively and work under the stronger assumption of access to a controlled unitary that prepares the state. 

Finally, we want to emphasize that the proposed reconstruction procedure can be empowered by advantageous measurement structure.
Storage-efficiency stems from the fact that we can keep track of the Hamiltonian -- not the associated Gibbs state -- which inherits structure from the underlying measurement procedure. Runtime savings are achieved by only exponentiating the Hamiltonian approximately and exploiting fast matrix-vector multiplication. 
We refer to App.~\ref{sec:classicalpostprocessing}
for details and content ourselves here with a vague, but instructive, analogy: View \textsc{Hamiltonian Updates} (Algorithm~\ref{alg:motivation}) as an adaptive cool-down procedure. We start with a Gibbs state at infinite temperature and, at each step, we cool down the system in a controlled fashion that guides the thermal state towards the unknown target. Importantly, each update is small and the number of total cooling steps is also benign. Hence, we never truly leave the moderate temperature regime and avoid computational bottlenecks that typically only arise at low temperatures. In turn, the output of our algorithm is in the form of a Hamiltonian whose Gibbs state is  close to the target state. 
A list of Gibbs state eigenvalues and corresponding eigenvectors can be obtained by block Krylov iterations, see App.~\ref{sec:convertingfromgibbstousual}. Runtime and memory cost of this conversion procedure can never exceed those of Algorithm~\ref{alg:motivation}.

\section{Numerical experiments}\label{sec:numerical}

\begin{figure}
  \centering
  \includegraphics[width=0.5\columnwidth]{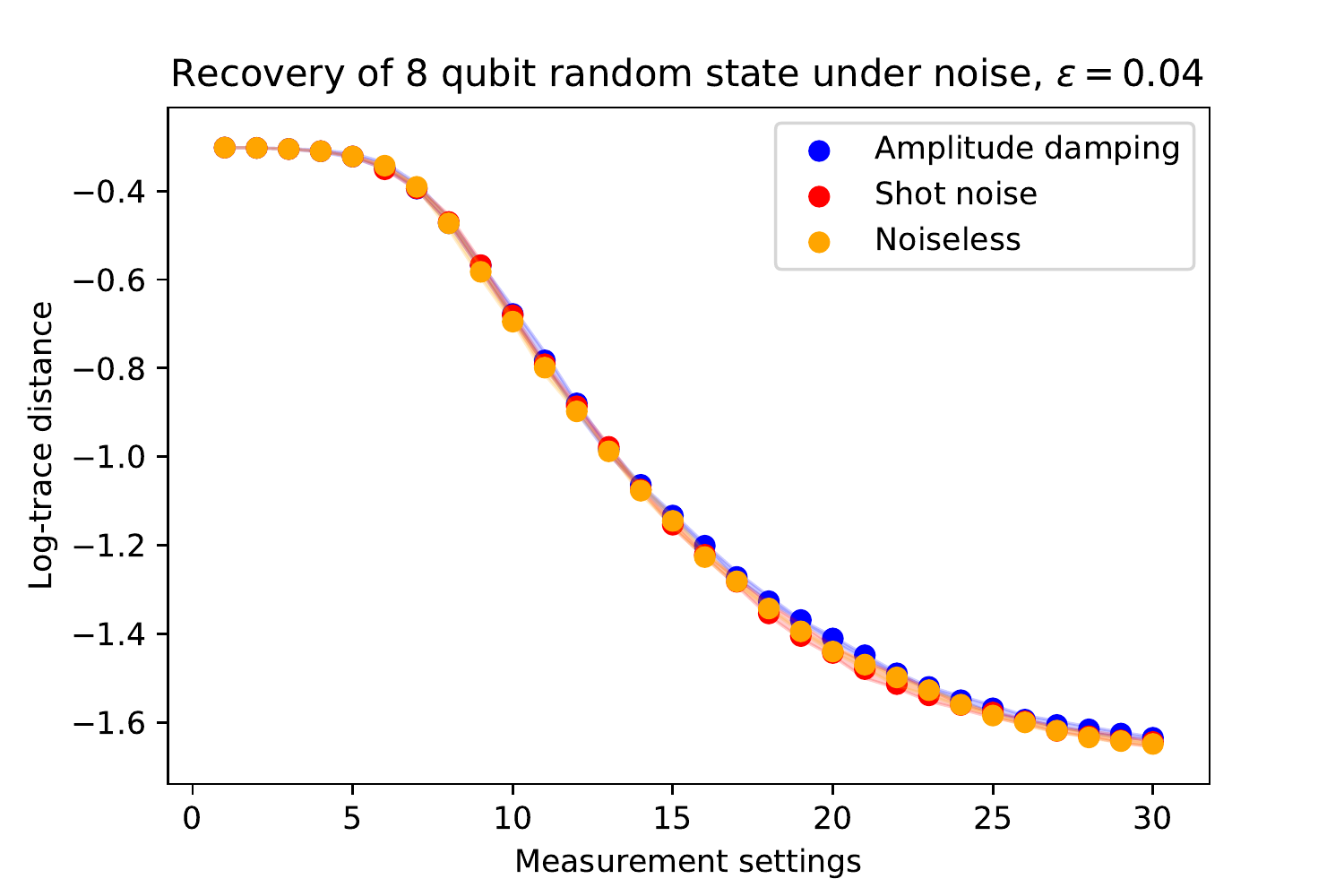}
\caption{
\emph{Convergence of Algorithm~\ref{alg:motivation} for different noise models.} 
We consider Haar-random global  measurements of a 8-qubit pure target state with target accuracy
$\epsilon =0.04$. 
Different colors track convergence for different noise models: (blue) amplitude damping noise with parameter $\epsilon/4$;
(red) white noise with standard deviation $\epsilon/4$ that mimics one-shot noise; 
(orange) zero noise. All logarithms are base 10 and the shaded area indicate $25\%$ and $75\%$ quartiles, estimated from 20 samples.
}
\label{fig:noise}
\end{figure}

We complement our theoretical assertions with empirical test evaluations for systems comprised of up to $12$ qubits. The results look promising and may establish Algorithm~\ref{alg:motivation} as a practical tool for quantum state tomography. We remark that our numerical implementation has two additional details when compared with the one described in Algorithm~\ref{alg:motivation}. Although these modifications do not change the asymptotic runtime analysis of the algorithm, they can substantially reduce runtime and sample complexity in practice.

The first alteration we do is to recycle the last measurement data after a successful update. More precisely, after each update $\sigma_t\to\sigma_{t+1}$, we then check if the new iteration $\sigma_{t+1}$ is still distinguishable from $\rho$ under the previous measurement basis. Only if this is not the case, we move on to sample a new measurement setting. 
Otherwise, we re-use the already known measurement basis to drive another update in the same direction.
We observe empirically that this minor modification has very desirable consequences. It leads to a much faster convergence throughout early stages of the algorithm and, by extension, reduces the number of required measurement settings significantly. 

What is more, this recycling procedure cannot change the asymptotic scaling of the algorithm. To see this, note that the modification can only affect postprocessing complexity. Indeed, it clearly does not require us to sample more states or measurement settings. Finding another violation can only bring us closer to the state in relative entropy. And the postprocessing time can only double in the worst case. This worst case scenario happens when after updating every basis once, we have already converged in that basis and checking again does not lead to further convergence. We will refer to this variation as the \emph{last step recycling strategy}. It is explained in detail in the appendix (Algorithm~\ref{alg:HUtomo_last}).

Other variations of this basic principle come to mind. For instance, we need not stop at testing the current iteration against the previous measurement basis. We can also test it against all measurements that have already accumulated. This variation can further reduce the (total) number of basis settings required to converge. Fig.~\ref{fig:cs15} confirms this intuition. 
 However, this strategy comes at the expense of an increase in the computational complexity of the post-processing. We refer to this strategy as the \emph{complete recycling strategy}.
 
 \begin{figure}
\centering
\begin{tabular}{lr}
\includegraphics[width=0.5\columnwidth]{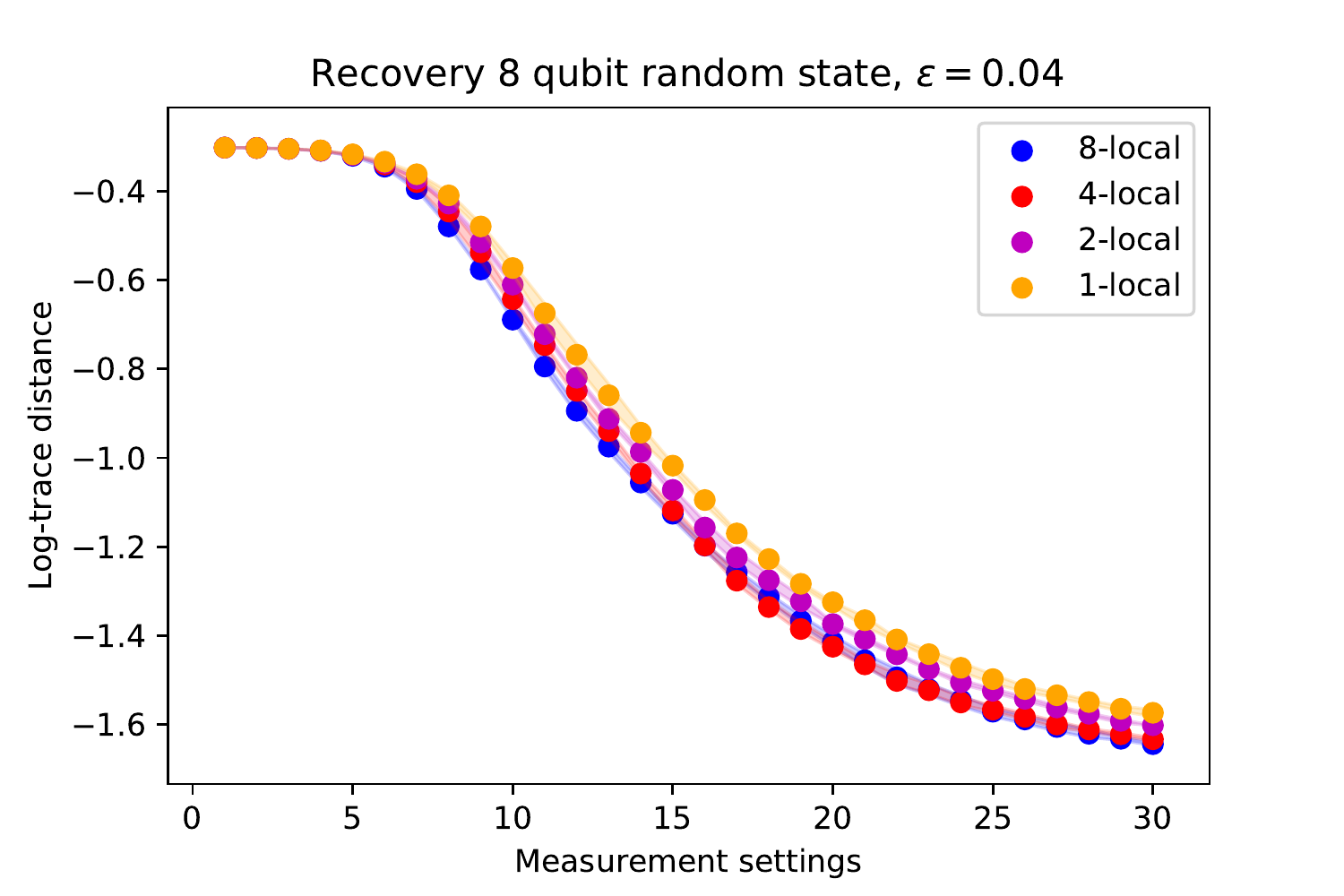}
&
  \includegraphics[width=0.5\columnwidth]{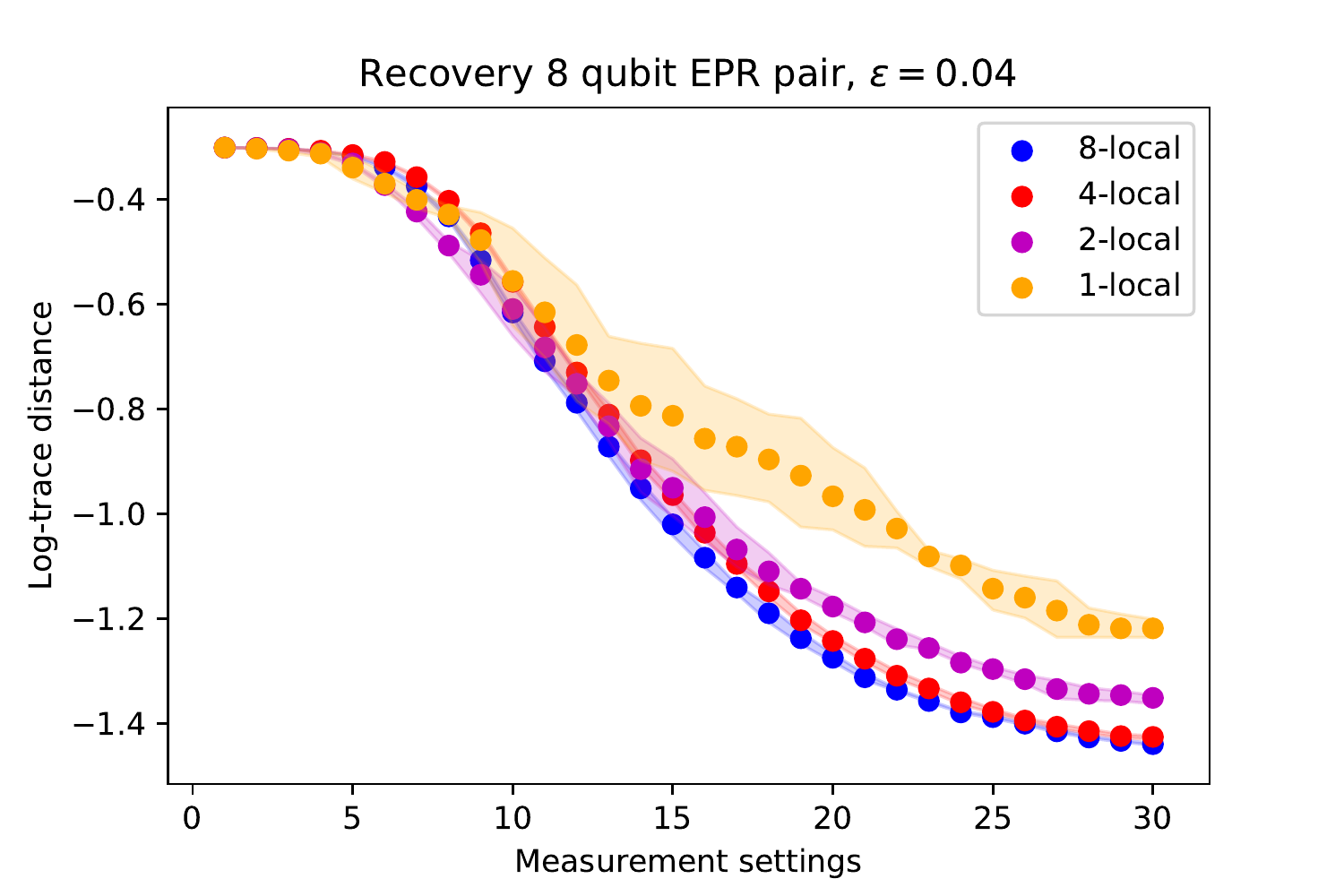}
\end{tabular}
\caption{
\emph{Convergence of Algorithm~\ref{alg:motivation} for different measurement localities.}
Different colors track convergence (in logarithmic trace distance) for 8-qubit  basis measurements with different localities and target accuracy $\epsilon=0.04$. Individual basis measurements are subject to white noise with standard deviation $\epsilon/4$.
(Left) Reconstruction of a generic pure target state. (Right) Reconstruction of a  highly structured target state (EPR/Bell state).
All logarithms are base 10 and the shaded area indicate $25\%$ and $75\%$ quartiles, estimated from 20 samples.
\label{fig:random+epr}
}
\end{figure}

Apart from these practical improvements, we have also tested desirable fundamental properties of Algorithm~\ref{alg:motivation}. Chief among them is noise resilience. As advertised in Sec.~\ref{sec:overview} and proved in App.~\ref{sec:robust}, the performance of the algorithm under arbitrary noise of bounded intensity is indistinguishable from the noiseless case. This feature is empirically confirmed by Fig.~\ref{fig:noise}. For detecting a random pure state on 8 qubits, different noise sources -- such as shot noise and amplitude damping -- affect convergence in a very mild fashion only (\emph{robustness}). 
It is also interesting to note that the convergence in trace norm appears to be polynomial for the first measurements and then switches to an exponential phase.

Another interesting figure of merit is measurement locality. 
The assertions that underpin Algorithm~\ref{alg:motivation} do, in principle, extend to local measurement primitives. But, as detailed in App.~\ref{sub:local-measurements}, the resulting numbers look rather pessimistic and scale unfavorably with measurement locality $k$.
Empirical studies do paint a much more favorable picture, see Fig.~\ref{fig:random+epr}.
The two subplots address reconstruction of a typical 8-qubit target state (left), as well as a highly structured one (right).
A direct comparison lends credence to a conjecture voiced in App.~\ref{sub:local-measurements} below: generic or typical states are easier to reconstruct with local measurements than highly structured ones. We intend to address this gap between worst-case and average-case performance in future work.

\begin{figure}[htp!]
  \label{fig:epr10}
\centering
\begin{tabular}{cc}
\includegraphics[width=0.5\columnwidth]{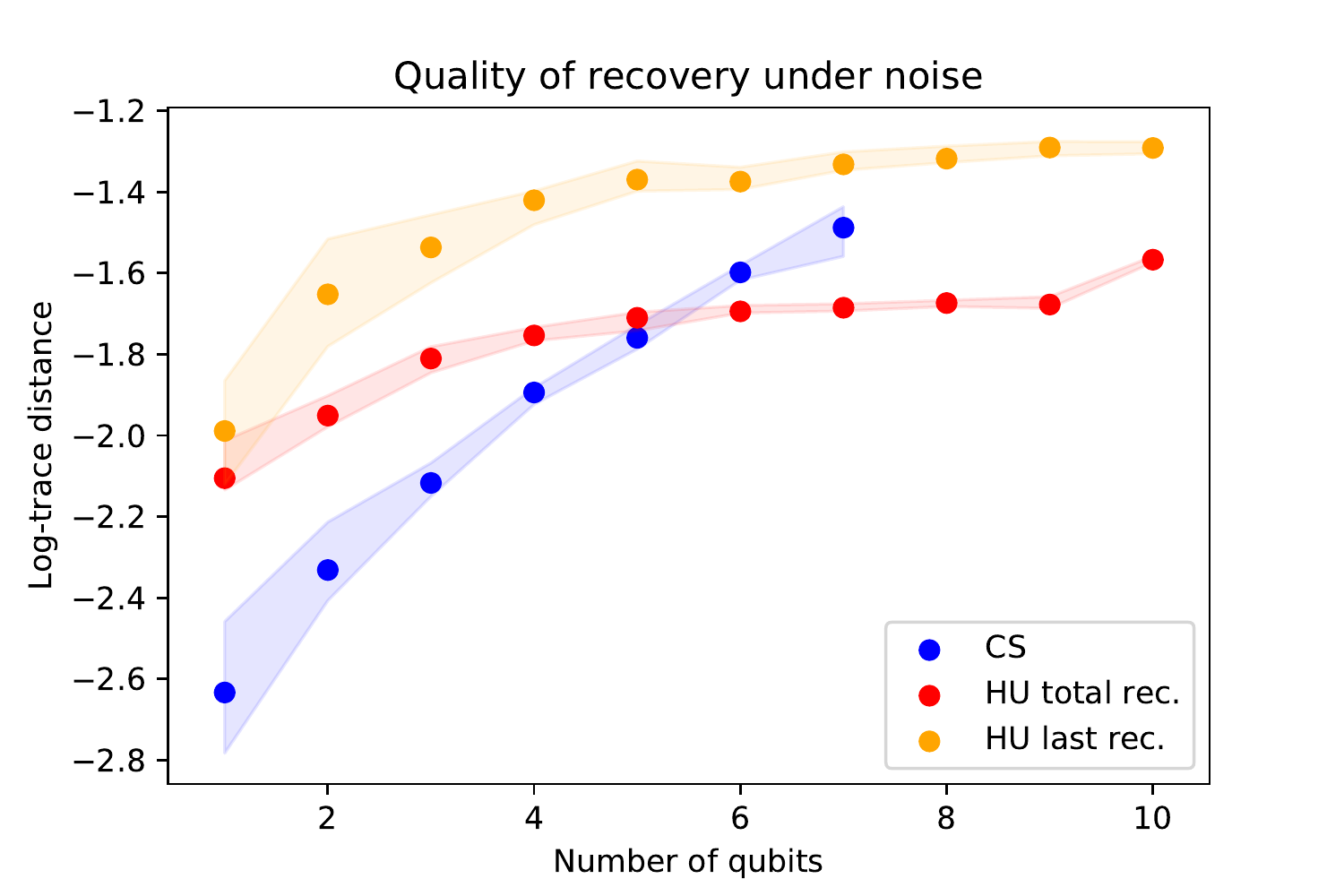} &
\includegraphics[width=0.5\columnwidth]{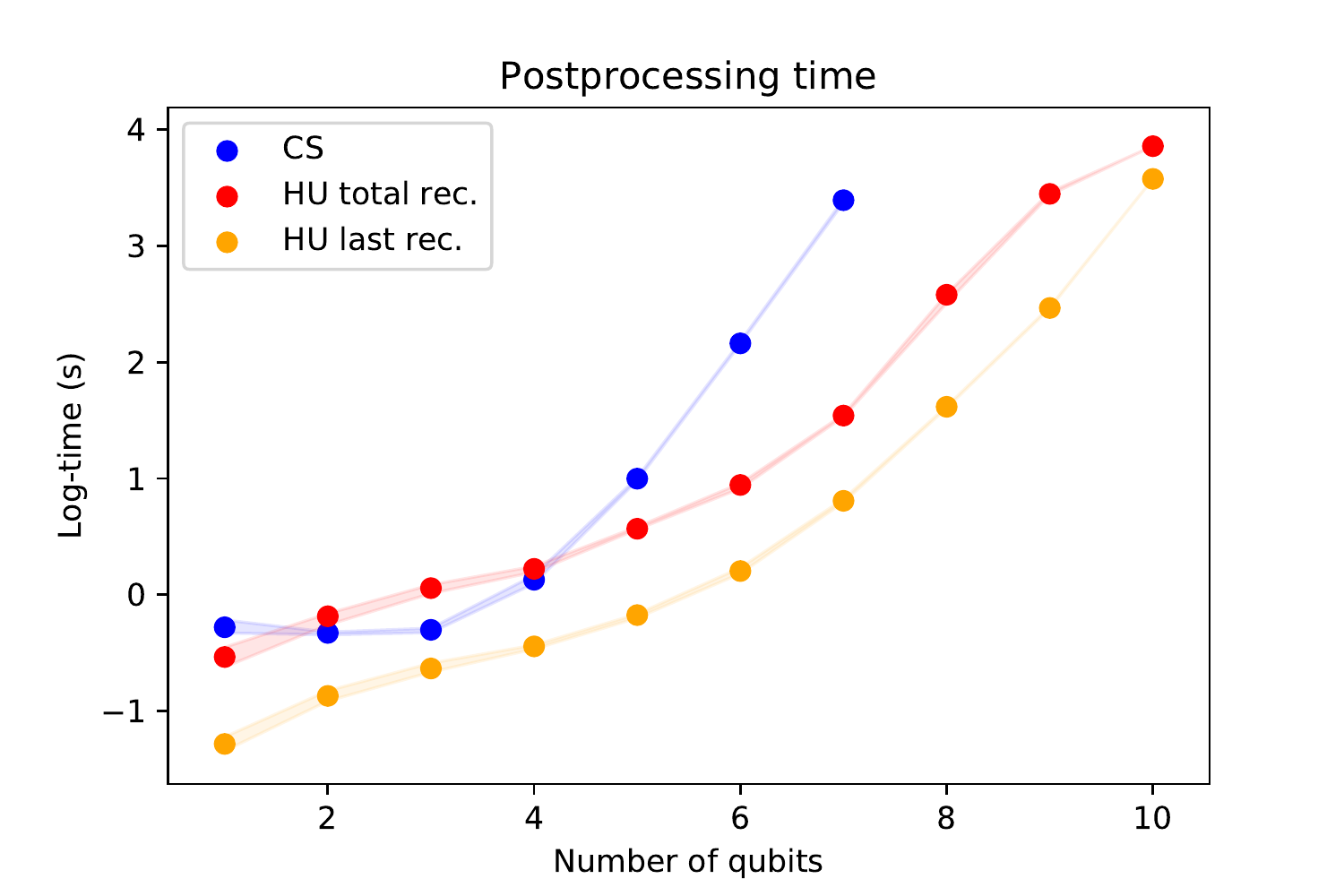}
\end{tabular}
\caption{\emph{Comparison between Algorithm~\ref{alg:motivation} and compressed sensing (CS) tomography.}
(Left) Reconstruction of a random $n$-qubit pure state from 15 globally random basis measurements corrupted by amplitude damping noise ($p=0.005$). Different colors track the logarithmic trace distance error achieved by either compressed sensing (blue) or variants of Algorithm~\ref{alg:motivation} (orange and red) for $\epsilon=0.01$. Shaded regions indicate the $25-75$ percentiles over 20 independent runs.
(Right) Empirical runtime for executing (naive implementations of) the three different reconstruction procedures on a conventional laptop. CVX \cite{cvx} --  a standard solver for semidefinite programs -- could not go beyond $7$ qubits. 
}
\label{fig:cs15}
\end{figure}

Last but not least, we compare Algorithm~\ref{alg:motivation} against the state of the art regarding tomography from very few basis measurements. Compressed sensing \cite{Gross2010,Flammia2012,kueng2015,kueng_low_2017} has  been designed to fit a low rank solution to the observed measurement data by also minimizing the nuclear norm over the cone of positive semidefinite matrices:
\begin{align}
\text{minimize}_{X \succeq 0} \quad  \mathrm{tr}(X) \quad \text{subject to} \quad \quad \sum\nolimits_{i=1}^M \|\hat{p}_{U_i}(\rho) - p_{U_i}(X) \|_{\ell_2}^2 \leq \epsilon. \label{eq:cs}
\end{align}
Fig.~\ref{fig:cs15} compares Algorithm~\ref{alg:motivation} with compressed sensing (CS).
CS is contingent on solving a semidefinite program. We used CVX~\cite{cvx}, a standard SDP solver, in Python. Algorithm~\ref{alg:motivation} has also been implemented in Python. 
Open source code is available at~\cite{dsfranca2020}.
We see that Hamiltonian Updates is more noise-resilient than CS. The rightmost plot also underscores the importance of memory improvements.
A high-end desktop computer already struggles to solve SDP~\eqref{eq:cs} for $8$ qubits (even though the extrapolated computation time Fig.~\ref{fig:cs15} still seems reasonable),
while 10 qubits (and more) have not been a problem for Algorithm~\ref{alg:motivation}.
We believe that Fig.~\ref{fig:cs15} conveys both quantitative and qualitative advantages of Hamiltonian Updates over CS methods. This seems particularly noteworthy, because we compared both procedures for pure target states ($\mathrm{rank}(\rho)=1$) -- a use-case tailor-made for CS approaches.
We also stress that the implementation of the algorithm used to generate this data was not optimized, there is room for further improvements.

Let us conclude with the most important take-away from Figs.~\ref{fig:noise},~\ref{fig:random+epr} and~\ref{fig:cs15}. The theoretical assertions from Sec.~\ref{sec:overview} carry over to practice.  Moreover, recycling of data ensures that the number of measurement settings remains small even if we try to characterize the state up to high precision. 
Our theoretical results suggest that order $10^5$ algorithm iterations, and thus also measurement settings, might be required to obtain a $\epsilon=10^{-2}$-approximation of a pure state in dimension $D=2^{10}$. But our numerics demonstrate 
that already order $10^{1}$ suffice to achieve convergence. 
The main theoretical drawbacks of Algorithm~\ref{alg:motivation} -- most notably, the poor scaling in accuracy -- may be a non-issue in practical use cases. These findings establish our algorithm as a rare instance of a method that is provably (essentially) optimal and has a competitive performance in practice.

\paragraph*{Data and code availability.}
Source data and code are available for this paper~\cite{dsfranca2020}. All other data that support the plots within this paper and other findings of this study are available from the corresponding author upon reasonable request.

\paragraph*{Acknowledgements.}

We thank C.\ Ferrie, T.\ Grurl, C.\ Lancien,  R.\ K\"onig and J.A.\ Tropp for valuable input and helpful discussions.
F.B. and R.K. acknowledge funding from the US National Science Foundation (PHY1733907). 
The Institute for Quantum Information and Matter is an NSF Physics Frontiers Center.
D.S.F. acknowledges financial support from VILLUM FONDEN via the QMATH Centre of Excellence (Grant no. 10059).

\bibliographystyle{myalpha}
\bibliography{tomography}

\cleardoublepage

\appendix

\begin{figure}
    \centering
    \begin{tabular}{|c|c|c|}
    \hline
    \textbf{choose random basis} & \textbf{compare %
    statistics} & \textbf{update} \\
    \hline
    \hline
        \begin{tikzpicture}[baseline,scale=0.6]
    \fill[Blue,opacity=0.2] (0,0) circle(3);
    \draw (0,0) circle(3);
    \draw[fill=Blue] (0,3) circle (0.15);
    \draw[fill=Blue] (0,-3) circle (0.15);
    \draw[thick,Blue] (0,-3) -- (0,3);
    \draw[fill=Green] (45:3) circle(0.15);
    \node at (45:3.4) {\textcolor{Green}{$\rho$}};
    \draw[fill=Red] (0,0) circle (0.15);
    \node at (0.5,0) {\textcolor{Red}{$\sigma_0$}};
    \end{tikzpicture}
    &
    \begin{tikzpicture}[baseline=0.5cm,scale=0.6]
    \draw[thick,magenta, rounded corners, opacity=0.3] (0.25, -2.3) rectangle (1.75,3.25);
    \begin{scope}[yshift=1.5cm]
    \fill[Red] (-1.5,0) rectangle (-0.5,1);
    \fill[Red] (0.5,0) rectangle (1.5,1);
    \draw[thick,->] (-2,0)-- (2,0);
    \draw[thick,->] (-2,0) -- (-2,2);
    \node[rotate=90] at (-2.5,1) {{\small $\mathrm{Pr}[\pm | \textcolor{Red}{\sigma_0}]$}};
    \node at (-1,-0.5) {$+$};
    \node at (1,-0.5) {$-$};
    \end{scope}
    \begin{scope}[yshift=-1.5cm]
    \fill[Green] (-1.5,0) rectangle (-0.5,2*0.85355339059);
    \fill[Green] (0.5,0) rectangle (1.5,2-2*0.85355339059);
    \draw[thick,->] (-2,0)-- (2,0);
    \draw[thick,->] (-2,0) -- (-2,2);
    \node[rotate=90] at (-2.5,1) {{\small $\mathrm{Pr}[\pm | \textcolor{Green}{\rho}]$}};
    \node at (-1,-0.5) {$+$};
    \node at (1,-0.5) {$-$};
    \end{scope}
    \end{tikzpicture}
    &
    \begin{tikzpicture}[baseline,scale=0.6]
    \fill[Blue,opacity=0.2] (0,0) circle(3);
    \draw (0,0) circle(3);
     \draw[fill=Red,opacity=0.3] (0,0) circle (0.15);
    \node[opacity=0.3] at (0.5,0) {\textcolor{Red}{$\sigma_0$}};
    \draw[thick,opacity=0.3,->] (0,0.15) -- (0,3*0.462-0.15);
    \draw[fill=Green] (45:3) circle(0.15);
    \node at (45:3.4) {\textcolor{Green}{$\rho$}};
    \draw[fill=Red,opacity=0.3] (0,0) circle (0.15);
    \node[opacity=0.3] at (0.5,0) {\textcolor{Red}{$\sigma_0$}};
    \draw[fill=Red] (0,3*0.462) circle (0.15);
    \node at (0.5,3*0.462) {\textcolor{Red}{$\sigma_1$}};
    \end{tikzpicture} \\
    \hline
    $\left\{|+\rangle,|-\rangle \right\}$-basis &
    $\mathrm{Pr} \left[ - | \textcolor{Red}{\sigma_0} \right] > \mathrm{Pr} \left[ - |\textcolor{Green}{\rho}\right]$
    & $\sigma_1 \propto \exp \left( - \eta |- \rangle \! \langle -| \right)$
 \\
 \hline
 \hline
 
       \begin{tikzpicture}[baseline,scale=0.6]
      \fill[white] (0,-3) circle (0.15);
    \fill[Blue,opacity=0.2] (0,0) circle(3);
    \draw (0,0) circle(3);
    \draw[thick,Blue,<->] (-3,0) -- (3,0);
    \draw[fill=Blue] (-3,0) circle (0.15);
     \draw[fill=Blue] (3,0) circle (0.15);
    \draw[fill=Green] (45:3) circle(0.15);
    \node at (45:3.4) {\textcolor{Green}{$\rho$}};
    \draw[fill=Red] (0,3*0.462) circle (0.15);
    \node at (0.5,3*0.462) {\textcolor{Red}{$\sigma_1$}};
    \end{tikzpicture} 
    &
    \begin{tikzpicture}[baseline=0.5cm,scale=0.6]
    \draw[thick,magenta, rounded corners, opacity=0.3] (0.25, -2.3) rectangle (1.75,3.25);
    \begin{scope}[yshift=1.5cm]
    \fill[Red] (-1.5,0) rectangle (-0.5,1);
    \fill[Red] (0.5,0) rectangle (1.5,1);
    \draw[thick,->] (-2,0)-- (2,0);
    \draw[thick,->] (-2,0) -- (-2,2);
    \node[rotate=90] at (-2.5,1) {{\small $\mathrm{Pr}[0/1 | \textcolor{Red}{\sigma_1}]$}};
    \node at (-1,-0.5) {$0$};
    \node at (1,-0.5) {$1$};
    \end{scope}
    \begin{scope}[yshift=-1.5cm]
    \fill[Green] (-1.5,0) rectangle (-0.5,2*0.85355339059);
    \fill[Green] (0.5,0) rectangle (1.5,2-2*0.85355339059);
    \draw[thick,->] (-2,0)-- (2,0);
    \draw[thick,->] (-2,0) -- (-2,2);
    \node[rotate=90] at (-2.5,1) {{\small $\mathrm{Pr}[0/1 | \textcolor{Green}{\rho}]$}};
    \node at (-1,-0.5) {$0$};
    \node at (1,-0.5) {$1$};
    \end{scope}
    \end{tikzpicture}
    
    &
    
    \begin{tikzpicture}[baseline,scale=0.6]
    \fill[Blue,opacity=0.2] (0,0) circle(3);
    \draw (0,0) circle(3);
    \draw[fill=Green] (45:3) circle(0.15);
    \node at (45:3.4) {\textcolor{Green}{$\rho$}};
    \draw[fill=Red,opacity=0.3] (0,3*0.462) circle (0.15);
    \node[opacity=0.3] at (0.5,3*0.462) {\textcolor{Red}{$\sigma_1$}};
    \draw[opacity=0.3,->] (0+0.15,3*0.462) -- (3*0.43-0.15,3*0.43);
    \draw[fill=Red] (3*0.43,3*0.43) circle (0.15);
    \node at (3*0.43+0.5,3*0.43) {\textcolor{Red}{$\sigma_2$}};
    \end{tikzpicture}
    \\
    \hline
    
    $\left\{|0\rangle, |1 \rangle \right\}$-basis 
    &
    $\mathrm{Pr} \left[ 1| \textcolor{Red}{\sigma_1} \right] > \mathrm{Pr} \left[ 1| \textcolor{Green}{\rho} \right]$
    &
    $\sigma_2 \propto \exp \left( -\eta |- \rangle \! \langle-| - \eta |1 \rangle \! \langle 1|\right)$\\
    \hline
     \hline
 
       \begin{tikzpicture}[baseline,scale=0.6]
      \fill[white] (0,-3) circle (0.15);
    \fill[Blue,opacity=0.2] (0,0) circle(3);
    \draw (0,0) circle(3);
    \draw[thick,Blue,<->] (-3,0) -- (3,0);
    \draw[fill=Blue] (-3,0) circle (0.15);
     \draw[fill=Blue] (3,0) circle (0.15);
    \draw[fill=Green] (45:3) circle(0.15);
    \node at (45:3.4) {\textcolor{Green}{$\rho$}};
    \draw[fill=Red] (3*0.43,3*0.43) circle (0.15);
    \node at (3*0.43+0.5,3*0.43) {\textcolor{Red}{$\sigma_2$}};
    \end{tikzpicture} 
    &
    \begin{tikzpicture}[baseline=0.5cm,scale=0.6]
    \draw[thick,magenta, rounded corners, opacity=0.3] (0.25, -2.3) rectangle (1.75,3.25);
    \begin{scope}[yshift=1.5cm]
    \fill[Red] (-1.5,0) rectangle (-0.5,0.715*2);
    \fill[Red] (0.5,0) rectangle (1.5,0.285*2);
    \draw[thick,->] (-2,0)-- (2,0);
    \draw[thick,->] (-2,0) -- (-2,2);
    \node[rotate=90] at (-2.5,1) {{ \small $\mathrm{Pr}[0/1 | \textcolor{Red}{\sigma_2}]$}};
    \node at (-1,-0.5) {$0$};
    \node at (1,-0.5) {$1$};
    \end{scope}
    \begin{scope}[yshift=-1.5cm]
    \fill[Green] (-1.5,0) rectangle (-0.5,2*0.85355339059);
    \fill[Green] (0.5,0) rectangle (1.5,2-2*0.85355339059);
    \draw[thick,->] (-2,0)-- (2,0);
    \draw[thick,->] (-2,0) -- (-2,2);
    \node[rotate=90] at (-2.5,1) {{\small $\mathrm{Pr}[0/1 | \textcolor{Green}{\rho}]$}};
    \node at (-1,-0.5) {$0$};
    \node at (1,-0.5) {$1$};
    \end{scope}
    \end{tikzpicture}
    
    &
    
    \begin{tikzpicture}[baseline,scale=0.6]
    \fill[Blue,opacity=0.2] (0,0) circle(3);
    \draw (0,0) circle(3);
    \draw[fill=Green] (45:3) circle(0.15);
    \node at (45:3.4) {\textcolor{Green}{$\rho$}};
    \draw[fill=Red,opacity=0.3] (3*0.43,3*0.43) circle (0.15);
    \node[opacity=0.3] at (3*0.43+0.5,3*0.43) {\textcolor{Red}{$\sigma_2$}};
    \draw[opacity=0.3,->] (3*0.43+0.15,3*0.43) -- (3*0.722-0.15,3*0.360);
    \draw[fill=Red] (3*0.722,3*0.360) circle (0.15);
    \node at (3*0.722+0.5,3*0.360) {\textcolor{Red}{$\sigma_3$}};
    \end{tikzpicture}
    \\
    \hline
    
    $\left\{|0\rangle, |1 \rangle \right\}$-basis 
    &
    $\mathrm{Pr} \left[ 1| \textcolor{Red}{\sigma_2} \right] > \mathrm{Pr} \left[ 1| \textcolor{Green}{\rho} \right]$
    &
    $\sigma_3 \propto \exp \left( -\eta |- \rangle \! \langle -| - 2\eta |1 \rangle \! \langle 1|\right)$\\
    \hline
     \hline
 
       \begin{tikzpicture}[baseline,scale=0.6]
           \fill[Blue,opacity=0.2] (0,0) circle(3);
    \draw (0,0) circle(3);
    \draw[thick,Blue,<->] (0,-3) -- (0,3);
    \draw[fill=Blue] (0,-3) circle (0.15);
     \draw[fill=Blue] (0,3) circle (0.15);
    \draw[fill=Green] (45:3) circle(0.15);
    \node at (45:3.4) {\textcolor{Green}{$\rho$}};
 \draw[fill=Red] (3*0.722,3*0.360) circle (0.15);
    \node at (3*0.722+0.5,3*0.360) {\textcolor{Red}{$\sigma_3$}};
    \end{tikzpicture} 
    &
    \begin{tikzpicture}[baseline=0.5cm,scale=0.6]
    \draw[thick,magenta, rounded corners, opacity=0.3] (0.25, -2.3) rectangle (1.75,3.25);
    \begin{scope}[yshift=1.5cm]
    \fill[Red] (-1.5,0) rectangle (-0.5,0.680*2);
    \fill[Red] (0.5,0) rectangle (1.5,0.320*2);
    \draw[thick,->] (-2,0)-- (2,0);
    \draw[thick,->] (-2,0) -- (-2,2);
    \node[rotate=90] at (-2.5,1) {{\small$\mathrm{Pr}[\pm | \textcolor{Red}{\sigma_3}]$}};
    \node at (-1,-0.5) {$+$};
    \node at (1,-0.5) {$-$};
    \end{scope}
    \begin{scope}[yshift=-1.5cm]
    \fill[Green] (-1.5,0) rectangle (-0.5,2*0.85355339059);
    \fill[Green] (0.5,0) rectangle (1.5,2-2*0.85355339059);
    \draw[thick,->] (-2,0)-- (2,0);
    \draw[thick,->] (-2,0) -- (-2,2);
    \node[rotate=90] at (-2.5,1) {{\small $\mathrm{Pr}[+ | \textcolor{Green}{\rho}]$}};
    \node at (-1,-0.5) {$+$};
    \node at (1,-0.5) {$-$};
    \end{scope}
    \end{tikzpicture}
    
    &
    
    \begin{tikzpicture}[baseline,scale=0.6]
    \fill[Blue,opacity=0.2] (0,0) circle(3);
    \draw (0,0) circle(3);
    \draw[fill=Green] (45:3) circle(0.15);
    \node at (45:3.4) {\textcolor{Green}{$\rho$}};
     \draw[fill=Red,opacity=0.3] (3*0.722,3*0.360) circle (0.15);
    \node[opacity=0.3] at (3*0.722+0.5,3*0.360) {\textcolor{Red}{$\sigma_3$}};
    \draw[opacity=0.3,->] (3*0.722-0.05,3*0.360+0.15) -- (3*0.628+0.02,3*0.628-0.15);
    \draw[fill=Red] (3*0.628,3*0.628) circle (0.15);
    \node at (3*0.628+0.5,3*0.628-0.2) {\textcolor{Red}{$\sigma_4$}};
    \end{tikzpicture}
    \\
    \hline
    
    $\left\{|+\rangle, |- \rangle \right\}$-basis 
    &
    $\mathrm{Pr} \left[ -| \textcolor{Red}{\sigma_3} \right] > \mathrm{Pr} \left[ -| \textcolor{Green}{\rho} \right]$
    &
    $\sigma_4 \propto \exp \left( -2\eta |- \rangle \! \langle-| - 2\eta |1 \rangle \! \langle 1|\right)$\\
    \hline
        \end{tabular}
    \caption{\emph{Illustration of Algorithm~\ref{alg:motivation} for random Clifford measurements of a single-rebit state.} }
    \label{fig:illustration_rebit}
\end{figure}

\clearpage

\part*{Appendix}

\paragraph*{Roadmap}

Fig.~\ref{fig:illustration_rebit} provides a single-``rebit'' illustration of the proposed algorithm.
App.~\ref{sec:convergence} provides the convergence analysis of the algorithm and highlights how it relates to the number of required measurement settings and state copies.
App.~\ref{sec:examples} supplies 
concrete runtime bounds for different basis measurement primitives (4-design, Clifford, mutually unbiased bases and $k$-local 4-design).
Noise-robustness is established in App.~\ref{sec:robust}, while
App.~\ref{sec:classicalpostprocessing} explains how to perform classical postprocessing efficiently. A possible implementation on a quantum computer is provided in App.~\ref{sec:quantumimplementation}.
App.~\ref{sec:convertingfromgibbstousual} completes the postprocessing analysis (classical \& quantum) by providing a way to efficiently convert the algorithm output -- a Hamiltonian --  to a list of eigenvalues and eigenvectors.
Additional details and background can be found in  App.~\ref{sec:effective-rank} (effective rank) and App.~\ref{sec:fastvecmulti} (fast matrix-vector multiplication).

\section{Convergence analysis for Hamiltonian Updates}\label{sec:convergence}

In this section, we provide rigorous runtime and convergence guarantees for quantum state tomography with Hamiltonian Updates. This algorithm is based on mirror descent and a more detailed version of this algorithm is presented in Algorithm~\ref{alg:HUtomo_last}.
In order to understand convergence to the desired target state $\rho$, we need to specify a suitable distance measure.
Mirror descent with the von Neumann entropy as potential, and its cousins, quantify convergence in terms of the quantum relative entropy:
\begin{equation}
S(\rho \| \sigma) = \mathrm{tr} \left(\rho ( \log (\rho) - \log (\sigma) \right). 
\end{equation}
This choice of distance measure plays nicely with iterative updates inside a matrix exponential. Initialization with the maximally mixed state 
$\sigma_0 = \exp(0)/\mathrm{tr}(\exp(0))=\tfrac{1}{D}\mathbb{I}$ also begets an intuitive motivation. The relative entropy between (any) target $\rho$ and $\sigma_0$ is bounded by the logarithm of the ambient dimension:
\begin{equation}
S(\rho \|\sigma_0) \leq \log (D) \quad \text{for any state $\rho$}.
\label{eq:initialization}
\end{equation}
This is a suitable starting point. 
Hamiltonian Updates is designed to ensure that each iteration makes constant progress towards the target (in relative entropy).

\begin{lem} \label{lem:update-rule}
Fix a Hamiltonian $H_t$ and set $H_{t+1}=H_t+\eta P$, where $P$ is an orthoprojector and $\eta \in [0,1]$. Then, the Gibbs states $\sigma_t = \exp (-H_t)/\mathrm{tr}(\exp (-H_t))$ and $\sigma_{t+1} = \exp (-H_{t+1})/\mathrm{tr}(\exp(-H_{t+1})$ obey
\begin{equation*}
S(\rho \| \sigma_{t+1}) - S(\rho \|\sigma_t) \leq \eta \left( 2\eta + \mathrm{tr} \left( P (\rho-\sigma_t) \right) \right) \quad \text{for any state $\rho$.}
\end{equation*}
The r.h.s.\ is negative, provided that $\eta < \tfrac{1}{2}\mathrm{tr} \left( P (\sigma_t-\rho)\right)$.
\end{lem}

\begin{proof}
The matrix logarithm in relative entropies plays nicely with the matrix exponential associated with Gibbs states:
\begin{align}\label{equ:first-step}
S(\rho \| \sigma_{t+1}) - S (\rho \| \sigma_t) 
=& \mathrm{tr} \left( \rho (H_{t+1}-H_{t}) \right)
+ \log \left( \frac{\mathrm{tr}(\exp (-H_{t+1}))}{\mathrm{tr}(\exp (- H_t))} \right).
\end{align}
Let us bound the second term of Eq.~\eqref{equ:first-step}. The Peierls-Bogoliubov inequality states that
\begin{align}\label{equ:Peierls-Bogoliubov}
    &-\log\lb\frac{\mathrm{tr}(\exp (-H_{t}))}{\mathrm{tr}(\exp (- H_{t+1}))} \rb=-\log\lb\frac{\mathrm{tr}(\exp (-H_{t+1}+H_{t+1}-H_{t}))}{\mathrm{tr}(\exp (- H_{t+1}))}\rb\leq -\tr{(H_{t+1}-H_{t})\sigma_{t+1}}
\end{align}
Inserting Eq.~\eqref{equ:Peierls-Bogoliubov} into Eq.~\eqref{equ:first-step}, we get
\begin{align*}
    \mathrm{tr} \left( \rho (H_{t+1}-H_{t}) \right)
+ \log \left( \frac{\mathrm{tr}(\exp (-H_{t+1}))}{\mathrm{tr}(\exp (- H_t))} \right)\leq \mathrm{tr} \left( (\rho -\sigma_{t+1})(H_{t+1}-H_{t}) \right).
\end{align*}
It then follows that 
\begin{align*}
    \mathrm{tr} \left( (\rho -\sigma_{t+1})(H_{t+1}-H_{t}) \right)=\mathrm{tr} \left( (\rho -\sigma_{t})(H_{t+1}-H_{t}) \right)+\mathrm{tr} \left( (\sigma_t -\sigma_{t+1})(H_{t+1}-H_{t}) \right).
\end{align*}
Let us now bound the second term. \cite[Lem.~16]{Brandao2017a} implies 
$\tfrac{1}{2}\| \sigma_{t+1} - \sigma_t \|_{tr} \leq  \left( \exp (\eta \|P \|)-1\right) = \mathrm{e}^\eta-1 \leq  \eta+\eta^2 \leq 2 \eta$. Together with H\"older and inserting $H_{t+1}-H_{t}=\pm \eta P$ we then obtain
\begin{align*}
    \mathrm{tr} \left( (\sigma_t -\sigma_{t+1})(H_{t+1}-H_{t}) \right)\leq \frac{\eta}{2}\| \sigma_{t+1} - \sigma_t \|_{tr}\|P\|\leq 2\eta^2.
\end{align*}
By the definition of $H_{t+1}$ we further obtain
\begin{align*}
    \mathrm{tr} \left( (\rho -\sigma_{t})(H_{t+1}-H_{t}) \right)= \eta\mathrm{tr} \left( (\rho -\sigma_{t})P \right)
\end{align*}
and the claim follows.
\end{proof}

\begin{algorithm}[t]
\caption{\textit{Hamiltonian Updates for quantum 
state tomography with last step recycling.}
}
\begin{algorithmic}[H]
\State \textbf{Input:} error tolerance $\epsilon$,  number of loops $L$.
\State initialize: $t=0$, $H_t=0$, \textsc{convergence}=\textsc{false}
\While{\textsc{convergence}=\textsc{false}}
\State compute $\sigma_t = \exp(-H_t)/\mathrm{tr}(\exp(-H_t))$  \Comment current guess for the state $\rho$
\State select random basis measurement
$\left\{ U|i \rangle \! \langle i|U^\dagger \right\}$
\State compute %
outcome statistics $[p_i]$ of $\sigma_t$ 
 \Comment classical computation
\State estimate %
outcome statistics $[q_i]$ of $\rho$  \Comment quantum measurement
\State Set \textsc{Basis match}=\textsc{false}
\While{\textsc{Basis match}=\textsc{false}}
\State \textbf{check} if $[p_i]$ and $[q_i]$ are $\epsilon$-close in $\ell_1$ distance
\If{\textsc{no}}{ set
$P=\sum_{p_i>q_i} |i \rangle \! \langle i|$} 
\Comment{collect outcomes for which $p_i>q_i$}
\State Set $\eta=\frac{1}{8}\|p-q\|_{\ell_1}$
\State $H_{t+1} \gets H_t + \eta U^\dagger P U$ \Comment energy penalty for mismatch (in this basis)
\State update $\sigma_{t+1} = \exp(-H_{t+1})/\mathrm{tr}(\exp(-H_{t+1}))$ 
\State update outcome statistics $[p_i]$ of $\sigma_{t+1}$  \Comment recycling the measurement data
\State  $t\gets t+1$ \Comment{update number of updates counter}

\ElsIf{\textsc{yes}}\Comment current guess may be close to $\rho$
\State Set \textsc{Basis match}=\textsc{true} \Comment this basis does not provide updates anymore
\State check $L$ additional random bases \Comment{suppress likelihood of false positives}
\If {$\ell_1$ distance is always $<\epsilon$ } \Comment current guess is likely to be close
\State{\textbf{set} \textsc{convergence=true}}
\EndIf
\EndIf
\EndWhile
\EndWhile

\State \textbf{Output:} $H_t$
\end{algorithmic}
\label{alg:HUtomo_last}
\end{algorithm}

Lemma~\ref{lem:update-rule} ensures that every successful iteration in Algorithm~\ref{alg:HUtomo_last} makes constant progress towards the target, provided that $2 \eta$ (step size) is smaller than the observed measurement outcome discrepancy.
For $\left[q_i\right] = \langle i| U \rho U^\dagger |i \rangle$ and $\left[p_i \right] = \langle i| U \sigma U^\dagger |i \rangle$, the construction of the projector ensures
\begin{equation*}
\mathrm{tr}(U^\dagger PU (\rho-\sigma_t)) =
\sum_{p_i >q_i}(p_i - q_i ) =
\tfrac{1}{2}\sum_{i}|p_i - q_i | \leq \tfrac{\epsilon}{2}.
\end{equation*}
Combined with Eq.~\eqref{eq:initialization}, this readily implies a worst-case bound on the maximum number of iterations.

\begin{prop} \label{prop:convergence}
Fix a desired target accuracy $\epsilon$ and set $\eta=\epsilon/8$ (constant step size).
Then,
Algorithm~\ref{alg:HUtomo_last} terminates after at most $T=\lceil32\log (D)/\epsilon^2 \rceil$ steps. 
\end{prop}

\begin{proof}
For the sake of this argument, we assume that the algorithm doesn't terminate prematurely. 
The choice of step size together with Lemma~\ref{lem:update-rule} ensures that the $T$th iterate in Algorithm~\ref{alg:HUtomo_last} obeys
\begin{equation*}
S(\rho \|\sigma_T) - S(\rho \|\sigma_0)
=\sum_{t=0}^{T-1} \left( S(\rho \| \sigma_{t+1}) - S(\rho\| \sigma_t ) \right) \leq T \eta \left( 2 \eta - \tfrac{1}{2}\epsilon \right) = - \tfrac{\epsilon^2}{32}T.
\end{equation*}
Combined with Eq.~\eqref{eq:initialization} this implies
\begin{equation}
S(\rho \|\sigma_T) \leq S(\rho \|\sigma_0) - \tfrac{\epsilon^2}{32}T \leq \log (D) - \tfrac{\epsilon^2}{32}T.
\label{eq:convergence-aux}
\end{equation}
For $T \geq \lceil \frac{32 \log(D)}{\epsilon^2} \rceil$, the r.h.s.\ becomes negative -- an apparent contradiction to the nonnegativity of quantum relative entropy. 
To appropriately resolve this conflict, we need to take into account that each update in Algorithm~\ref{alg:HUtomo_last} is contingent on finding a basis measurement that is capable of distinguishing the current iterate $\sigma_t$ from $\rho$ (up to accuracy $\epsilon$). 
Viewed from this angle, Rel.~\eqref{eq:convergence-aux} simply states that it is impossible to find more than $T=\lceil \frac{32 \log(D)}{\epsilon^2} \rceil$ consecutive basis measurements that meet the update condition ($\sum_i |p_i -q_i| >\epsilon$). In other words: Algorithm~\ref{alg:HUtomo_last} terminates after at most $\lceil \frac{32 \log(D)}{\epsilon^2} \rceil$ steps.
\end{proof}

Proposition~\ref{prop:convergence} is an adaptation of standard convergence analysis arguments that is valid for any measurement primitive. It states that at some point, it becomes impossible to find \emph{any} new basis measurement that is capable of accurately distinguishing the current iterate $\sigma_T$ from the target. It does not address the problem of how to find suitable measurements and how one should actually check the current stage of convergence. 
Hamiltonian Updates is based on a simple routine to check both of them. Start with a basis measurement primitive $\mathcal{E}$ that is well-equipped for distinguishing the target $\rho$ from the current iterate $\sigma_t$. This is characterized by the ensemble parameters $\theta_\mathcal{E}(\rho)$ and $\tau_{\mathcal{E}}(\rho)$ defined in Eq.~\eqref{eq:tomo-sound-intro}. 
Sample $L$ basis measurements at random and compare the outcome distributions. If we find a noticeable discrepancy ($\sum_i |p_i - q_i|>\epsilon$), we use this basis to perform an update. If all $L$ pairs of outcome distributions are $\epsilon$-close, we conclude that it is likely that the algorithm has converged and $\sigma_t$ is close to $\rho$.
This stopping condition is supported by a simple probabilistic argument based on Eq.~\eqref{eq:tomo-sound-intro}.  Suppose that the current iterate obeys $\|\rho-\sigma_t \|_2 \geq \epsilon/\theta_{\mathcal{E}}(\rho)$, i.e.\ convergence has not been achieved yet. Then, the probability of failing to detect this discrepancy with $L$ independently sampled basis measurements is bounded by $\left(1-\tau_{\mathcal{E}}(\rho)\right)^L$. 
This highlights that the size of the control loop $L$ exponentially suppresses the probability of a false positive in step $t$ of the algorithm. 
For $\delta \in [0,1]$, the explicit (and ensemble-dependent) choice 
\begin{equation*}
L = \lceil \log (T)\log (1/\delta)/\tau_{\mathcal{E}}(\rho) \rceil
\quad \text{ensures} \quad 
(1-\tau_{\mathcal{E}}(\rho))^L \leq \delta/T,
\end{equation*}
where $T=\lceil 32 \log (D)/\epsilon^2 \rceil$ is the bound on the maximum number of iterations from Proposition~\ref{prop:convergence}. A union bound over all (actual) steps then ensures that the probability of incurring at least one false positive throughout -- i.e.\ $\|\rho-\sigma_t \|_2 \geq \epsilon/\theta_{\mathcal{E}}(\rho)$, but we fail to detect this discrepancy with $L$ independent basis measurements -- is bounded by $\delta$. 
Taking the contrapositive of this assertion and combining it with Proposition~\ref{prop:convergence} -- the algorithm must terminate after at most $T$ steps --  completes the convergence analysis. 

\begin{prop} \label{prop:convergence-accuracy}
Fix an (unknown) target state $\rho$ and a basis measurement primitive with parameters $\theta_{\mathcal{E}}(\rho),\tau_{\mathcal{E}}(\rho)$, as well as accuracy $\epsilon$ and error probability $\delta$.
Then, choosing $L=\lceil \frac{\log (T)\log (1/\delta)}{\tau_{\mathcal{E}}(\rho)} \rceil$ for the size of the control loop in Algorithm~\ref{alg:HUtomo_last} ensures that 
 the output $\sigma_\star$ of Algorithm~\ref{alg:HUtomo_last} obeys
\begin{equation*}
\|\sigma_\star - \rho \|_2 \leq  \epsilon / \theta_{\mathcal{E}}(\rho).
\quad \text{with probability at least $1-\delta$.}
\end{equation*}
Here, $T=\lceil 32 \log (D)/\epsilon^2 \rceil$ denotes the upper bound on the maximum number of updates within Algorithm~\ref{alg:HUtomo_last}.
\end{prop}

This is the main technical result of this work. It establishes a probabilistic convergence guarantee for the output of Algorithm~\ref{alg:HUtomo_last}. Note that the established accuracy $\epsilon /\theta_{\mathcal{E}}(\rho)$ is worse than the original accuracy parameter (typically: $\theta_{\mathcal{E}}(\rho)<1$) and depends on the target state. What is more, Prop.~\ref{prop:convergence-accuracy} establishes closeness in Frobenius norm only.
We can convert it into a trace distance bound at the cost of an extra rank factor $r=\mathrm{rank}(\rho)$:
\begin{equation*}
\|\sigma_\star - \rho \|_{\mathrm{tr}} \leq \sqrt{r} \| \sigma_\star - \rho \|_2
\leq \sqrt{r} \epsilon/\theta_{\mathcal{E}}(\rho)
\end{equation*}

We refer to App.~\ref{sec:effective-rank} for a proof of this  conversion rule. Furthermore, it is possible to replace an assumption on the rank by a continuous relaxation thereof. More precisely, define the $\alpha-$effective rank of $\rho$ as
\begin{align*}
 r_{\text{eff},\alpha}(\rho)=\tr{\rho^{\alpha}} ^{\frac{1}{1-\alpha}} \quad \text{with $\alpha \in (0,1)$.}
\end{align*}
We refer to Appendix~\ref{sec:effective-rank} for a discussion of this quantity. We note that, up to exponentiation and normalization, it corresponds to the $\alpha-$R\'enyi entropy of $\rho$ and satifies $\lim_{\alpha\to 0}r_{\text{eff},\alpha}(\rho)=r(\rho)$.
In Cor.~\ref{cor:effectiverank} we show that
 \begin{align}\label{equ:conversionrule}
    \|\rho-\sigma\|_1\leq 2
r_{\mathrm{eff},\alpha}(\rho)^{\frac{1}{2}}\varepsilon^{-\frac{\alpha}{2(
\alpha-1)}}\|\rho-\sigma\|_2+2\varepsilon(1-\alpha)^{\frac{1}{\alpha}}.
\end{align}
for any $\epsilon>0$.
We conclude that Algorithm~\ref{alg:HUtomo_last} can be equipped with rigorous convergence guarantees in trace distance also -- albeit at the cost of an extra multiplicative factor in the original accuracy $\epsilon$. 
However, given a bound on $r_{\mathrm{eff},\alpha}$, this can be offset by running the algorithm with adjusted accuracy 
\begin{equation}\label{equ:adjustedepsilon}
\varepsilon=\theta_{\mathcal{E}}(\rho)r_{\mathrm{eff},\alpha}(\rho)^{-\frac{1}{2}}\epsilon^{1+\frac{\alpha}{2(1-\alpha)}}.
\end{equation}
This slight adjustment -- that only depends on the measurement primitive and the target (effective rank) -- gives rise to a convergence guarantee in trace distance. 

\begin{thm}[Detailed re-statement of Theorem~\ref{thm:main-informal}] \label{thm:main-detail}
Suppose that we wish to reconstruct a $D$-dimensional target state $\rho$ with rank $r$ up to accuracy $\epsilon$ in trace distance with probability at least $1-\delta$. 
Then, Hamiltonian Updates -- Algorithm~\ref{alg:HUtomo_last} -- based on any basis measurement primitive with parameters $\theta_{\mathcal{E}}(\rho),\tau_{\mathcal{E}}(\rho)>0$ achieves this goal, provided that we make the following parameter choices:
\begin{align*}
\varepsilon =& \theta_{\mathcal{E}}(\rho) \epsilon /\sqrt{r} & \text{(accuracy within the algorithm)},\\
\eta =& \varepsilon/8 = \theta_{\mathcal{E}}(\rho) \epsilon /(8\sqrt{r}) & \text{(step size)}, \\
T=& \lceil 32\log(D)/\varepsilon^2 \rceil = \lceil 32 \log (D) r/\theta_{\mathcal{E}}^2(\rho)\epsilon^{-2} \rceil & \text{(maximum number of iterations)}, \\
L=& \lceil \log(T) \log (1/\delta)/\tau_{\mathcal{E}}(\rho) \rceil %
& \text{(size of the control loops)}.
\end{align*}
This corresponds to at most
\begin{equation}
M = TL = \log (1/\delta) T \log (T)/\tau_{\mathcal{E}}(\rho) = \tilde{\mathcal{O}}\big(r(\rho)/(  \tau_{\mathcal{E}}(\rho) \theta_{\mathcal{E}}(\rho)^2 \epsilon^2) \big)
\label{eq:basis-cost}
\end{equation}
different measurement settings.
\end{thm}
We only stated the recovery guarantees in terms of the rank in order not to over-complicate the presentation. However it is straightforward to adapt the bounds for the effective rank with the aid of Eq.~\eqref{equ:adjustedepsilon} and Eq.~\eqref{equ:conversionrule}. We restate in full generality in Thm~\ref{thm:main-stable} of Appendix~\ref{sec:effective-rank}.
So far, we have not taken into account the effect of statistical fluctuations when estimating outcome distributions of the unknown state. 
\textsc{Hamiltonian Updates} is designed to be robust with respect to errors and noise. Estimating outcome distributions of the unknown state up to accuracy $\mathcal{O}(\varepsilon)$ is sufficient to drive progress within the algorithm. A total of $N_{\text{single basis}}=\mathcal{O}(D \varepsilon^{-2})=\mathcal{O}(D r(\rho)/(\theta_{\mathcal{E}}^2(\rho))\epsilon^2)$ samples per basis measurements suffice to meet this accuracy threshold. Thus, by suitably reducing the step size at each iteration, it is possible to account for both noise in the measurements and statistical fluctuations. This is discussed in more detail in App.~\ref{sec:robust}.

\begin{cor}[Worst-case sample complexity ] \label{cor:sample-complexity}
The total number of state samples required to execute the procedure detailed in Theorem~\ref{thm:main-detail} is at most
\begin{equation*}
N = N_{\text{single basis}}M 
= \tilde{\mathcal{O}} \left( r(\rho)^2 D / (\tau_{\mathcal{E}} (\rho)\theta_{\mathcal{E}}^4(\rho) \epsilon^4 )\right).
\end{equation*}
\end{cor}

\section{Concrete runtime bounds via quantum state distinguishability} \label{sec:examples}
 
Theorem~\ref{thm:main-detail} provides a rigorous convergence guarantee for Hamiltonian Updates (Algorithm~\ref{alg:HUtomo_last}).
This, in turn, also bounds the required number of basis measurements (see Rel.~\eqref{eq:basis-cost} and sample complexity (Corollary~\ref{cor:sample-complexity}).
These bounds all depend on parameters \eqref{eq:tomo-sound-intro} that capture how well the measurement primitive can distinguish pairs of states. 
This question has a long and rich history that dates back to Helstrom \cite{Helstrom1969} and Holevo \cite{Holevo1973}. These pioneering works showed that the optimal probability of correctly distinguishing two known states $\rho,\sigma$ is proportional to their trace distance: $p_{\mathrm{succ}}=\tfrac{1}{2}+\tfrac{1}{4}\|\rho-\sigma\|_1$. 
The optimal distinguishing measurement depends on the states in question (it is the projector onto the positive range of $\rho-\sigma$). Later, Ambainis and Emerson considered an interesting variation of the problem: What is the optimal probability of distinguishing two states with a \emph{fixed} measurement primitive? 
In this case, the measurement procedure is fixed and it is only possible to optimize the probability of success classically over the resulting outcome distributions. For the measurement primitive considered here -- unitary transformations $U \sim \mathcal{E}$ followed by a computational basis measurement -- the maximum likelihood rule yields
$
p_{\mathrm{succ}}=\frac{1}{2}+\frac{1}{4}\mathbb{E}_{U \sim \mathcal{E}} \|p_U (\rho) - p_U (\sigma) \|_{\ell_1}
$
which is optimal, see e.g.\ \cite{Matthews_2009}. The bias can be related to a distance in state space:
\begin{equation}
\mathbb{E}_{U \sim \mathcal{E}} \|p_U (\rho) - p_U (\sigma) \|_{\ell_1} \geq \lambda_{\mathcal{E}}(\rho,\sigma) \| \rho - \sigma \|_2.
\label{eq:distinguishing}
\end{equation}
The proportionality constant $\lambda_{\mathcal{E}} (\rho,\sigma)$ measures how well the measurement is equipped to distinguish $\rho$
 from $\sigma$. This constant is positive for every state pair if and only if the ensemble implements a tomographically complete measurement and
has been the subject of considerable attention \cite{ambainis_wise_2007,Matthews_2009,Lancien_2013,kueng_2016_distinguishing}. Tight bounds have been derived for a variety of measurement procedures. It should not come as a surprise that these bounds can be converted into statements about the ensemble parameters \eqref{eq:tomo-sound-intro} that govern the runtime of Hamiltonian Updates.

\begin{lem} \label{lem:conversion}
Fix two states $\rho,\sigma$ and suppose that a measurement primitive $\mathcal{E}$ obeys Rel.~\eqref{eq:distinguishing}, as well as $\mathbb{E}_{U \sim \mathcal{E}} \|p_U (\rho) - p_U (\sigma)\|_{\ell_2}^2 \leq \frac{1}{D} \|\rho-\sigma\|_2^2$.
Then, the following choice of ensemble parameters satisfies Rel.~\eqref{eq:tomo-sound-intro}:
$\theta_{\mathcal{E}}(\rho,\sigma) = \tfrac{1}{2} \lambda_{\mathcal{E}}(\rho,\sigma)$ and $\tau_{\mathcal{E}}(\rho,\sigma) = \tfrac{1}{4} \lambda_{\mathcal{E}}(\rho,\sigma)^2$.
\end{lem}

The proof is an immediate consequence of the Paley-Zygmund inequality. The extra assumption $\mathbb{E}_{U \sim \mathcal{E}} \|p_U (\rho) - p_U (\sigma) \|_{\ell_2}^2 \leq \tfrac{1}{D} \|\rho - \sigma \|_2^2$ is a mild anti-concentration condition. Most reasonable measurement primitives have this feature. 
The following subsections discuss two examples, one non-example and a possible extension to $k$-local measurements.

\subsection{Haar random unitaries and approximate 4-designs}

Let us start with the most generic measurement primitive conceivable: each $U$ is a random unitary that is selected according to the unique unitarily invariant (Haar) measure on the full $D$-dimensional unitary group. 
Although impractical, this measurement model lends itself to a thorough mathematical analysis. Haar integration is a powerful technique that allows for computing (even) moments of the measurement outcome distribution -- regardless of the states $\rho$ and $\sigma$ in question. 
Ambainis and Emerson \cite{ambainis_wise_2007} used this feature to infer
\begin{align}
\mathbb{E}_{\text{Haar}} \left[ \|p_U (\rho) - p_U (\sigma) \|_{\ell_1} \right] \geq & \tfrac{1}{3}\|\rho-\sigma\|_2
\quad
\text{and} \nonumber \\
\mathbb{E}_{\text{Haar}} \left[\|p_U(\rho) - p_U (\sigma) \|_{\ell_2}^2\right] =& \tfrac{1}{D+1}\|\rho-\sigma \|_2^2, \label{eq:4-design}
\end{align}
see also \cite{Matthews_2009}.
Remarkably, the first relation follows from combining information about the second and fourth moment only \cite{Berger1997}, while the second relation \emph{is} the second moment. 
Thus, any measurement ensemble that reproduces the first four moments of the Haar random measurement primitive obeys the same relations.
Unitary ensembles with this property are known as (unitary) 4-designs \cite{Dankert2009,Gross2007} and have been identified as a versatile tool in quantum information. Applications range from partially de-randomizing quantum information protocols to the study quantum chaos and complexity. 
While exact 4-designs are notoriously difficult to construct, several approximate constructions are known \cite{Brand_o_2016,Onorati2017,Hunter-Jones2019,Haferkamp2020}. 
For instance, for $n$-qudit systems ($D=d^n$), 
local random circuits of size $T=\mathcal{O} \left( n^2 \right)$ approximate the first four moments of the Haar measure sufficiently accurately to ensure Rel.~\eqref{eq:4-design} \cite{Brand_o_2016}. Other, more recent results yield qualitatively similar results~\cite{Onorati2017,Hunter-Jones2019,Haferkamp2020}.
Lemma~\ref{lem:conversion} then allows us to convert this insight into bounds on the ensemble parameters \eqref{eq:tomo-sound-intro} associated with an (approximate) 4-design: $\theta_{\text{4-design}}=\frac{1}{6}$ and $\tau_{\text{4-design}}=\frac{1}{36}$.
Both ensemble parameters are constant and do not affect the scaling of Hamiltonian Updates significantly. The number of required measurement settings $M_{\text{4-design}}$ and state copies $N_{\text{4-design}} $ amount to
\begin{align*}
M_{\text{4-design}} = \mathcal{O} \left( r \log (D)/\epsilon^2 \right) \quad \text{and} \quad
\quad N_{\text{4-design}} = \mathcal{O} \left( D r^2 \log (D)/\epsilon^4 \right).
\end{align*}
These numbers saturate fundamental lower bounds up to $\log(D)$ and $1/\epsilon^2$. 
The fact that (approximate) 4-designs can be realized by local quantum circuits of size $\mathcal{O}(n^2)$ also has profound implications for storage and runtime. 
We refer to Table~\ref{tab:comparison} for an illustration and note that, to the best of our knowledge, we outperform all existing protocols in the postproceessing. Regarding storage, substantial savings can be achieved by not storing the unitary -- a dense $d^n \times d^n$ matrices -- themselves, but their circuit diagrams -- collections of $O(n^2)$ $4 \times 4$ matrices. As we can run our algorithm storing only the unitaries and measurement outcomes, this implies we can run our algorithm only requiring $\cO(Dr\log(D)\epsilon^{-2})$ memory.
Runtime savings hail from the insight that compact circuit diagram descriptions do imply a fast $\tcO(D)$ matrix-vector multiplication for the underlying unitaries. The canonical example is the fast Fourier transform, but the principle applies more broadly. This yields a total runtime of $\tcO( D^2r^{\frac{5}{2}}\epsilon^{-5})$. 
We refer to App.~\ref{sec:classicalpostprocessing} for details.

\begin{table*}
\begin{center}
{\scriptsize
\begin{tabular}{|c|c|c|c|c|c|}
\hline
   & \textbf{meas.\ primitive}  & \textbf{basis settings} & \textbf{state copies}  & \textbf{runtime} & \textbf{memory} \\
\hline
\hline
 CS \cite{Flammia2012}
 & Pauli observables
 & $D r$
 & $D^2 r^2 \epsilon^{-2}$ & $D^4$ & $D^3$ \\
 \hline
PLS \cite{guta_fast_2018}
 & Pauli bases & $D^{1.6}$ & $D^{1.6}r^2 
 \epsilon^{-2}$ & $D^3$ & $D^2$  \\
 \hline
 this work & $k$-local 4-design & $D^{8.33/k} r \epsilon^{-2}$ & $D^{1+12.5/k} r^2 \epsilon^{-4}$   & $D^{2+8.33/k}r^{5/2}\epsilon^{-5}$ & $D^{8.33/k+1}r\epsilon^{-2}$ \\
 \hline
\end{tabular}
}
\end{center}
\caption{\emph{Resource scaling for state tomography protocols based on local measurements (single copy):} Here, $D$ denotes the Hilbert space dimension, $r$ is the rank of the target state and $\epsilon$ is the desired target accuracy (in trace distance). We have suppressed constants, as well as logarithmic dependencies in $D$ and $r$.}
\label{tab:local}
\end{table*}

\subsection{Clifford unitaries}

Let us focus on quantum systems comprised of $n$-qubits, i.e.\ $D=2^n$. 
The Clifford group is the collection of all possible quantum circuits that can be generated by \textsc{CNOT}, Hadamard and $\pi/4$-phase gates only. It has an extensively rich and well-understood structure \cite{Gottesman1997} and it is widely believed that Clifford unitaries are easier to implement than general quantum circuits of size $\mathcal{O}(n^2)$ -- like approximate 4-designs.
Moreover, the stabilizer formalism allows for storing Clifford unitaries very efficiently; while the development of a fast matrix-vector multiply is also possible.
These desirable features motivate the adoption of a Clifford measurement primitive for quantum state tomography. Regarding the theoretical analysis of Hamiltonian updates (and quantum state distinguishability), the transition from approximate 4-designs to Clifford unitaries is not completely straightforward, however. 
The Clifford group does constitute a 3-design \cite{Zhu2017,KG2015,Webb2015}, but not a 4-design \cite{zhu2016clifford}.
While this feature implies $\mathbb{E}_{U \sim \mathrm{Clifford}}\|p_U (\rho)-p_U (\sigma) \|_{\ell_2}^2 =\tfrac{1}{D+1}\|\rho-\sigma\|_2^2$,
the essentially optimal distinguishability bound from \eqref{eq:4-design} does not apply in general. 
It has to be replaced by
\begin{equation*}
\mathbb{E}_{U \sim \mathrm{Clifford}} \left[ \|p_U (\rho) - p_U (\sigma) \|_{\ell_1}\right] \geq %
\frac{\|\rho-\sigma \|_2}{4 \sqrt{\mathrm{rank}(\rho)}}  \quad \text{for any state $\sigma$,}
\end{equation*}
see \cite{kueng_2016_distinguishing}.
Although weaker than its 4-design counterpart, this rank-dependent scaling is unavoidable and does affect the ensemble parameters: $\theta_{\mathrm{Clifford}}(\rho)=\frac{1}{8 \sqrt{r}}$ and $\tau_{\mathrm{Clifford}}(\rho)=\frac{1}{64 r^2}$, with $r=\mathrm{rank}(\rho)$.
These in turn control the worst-case performance of Hamiltonian Updates in terms of measurement settings and sample complexity:
\begin{align*}
M_{\mathrm{Clifford}}=& \mathcal{O} \left( r^3 \log (D)/\epsilon^2 \right) \quad \text{and} \\
N_{\mathrm{Clifford}}=& \mathcal{O} \left( D r^5 \log(D)/\epsilon^4 \right).
\end{align*}
Although worse than their 4-design counterpart, these assertions are still essentially optimal in the low-rank regime.
In particular, the required number of measurement settings is only logarithmic in the ambient dimension.
While this is a substantial improvement over existing results \cite{kueng2015}, the overall resource count is still considerably larger than the 4-design case. One way to overcome this discrepancy is to interleave random Clifford rotations with (few) single-qubit non Clifford.
This modification upgrades random Clifford circuits to an approximate 4-design \cite{Haferkamp2020}.

\subsection{Mutually unbiased bases (non-example)}

Two orthonormal bases $\left\{|b_i\rangle:i=1,\ldots,D \right\}$ and $\left\{|c_j \rangle:j=1,\ldots, D\right\}$ of $\mathbb{C}^D$ are \emph{mutually unbiased} if $\left| \langle b_i, c_j \rangle \right|^2 = \frac{1}{D}$ for all $i,j=1,\ldots,D$. Standard and Fourier basis are the prototypical example, but there are many others. At most $(D+1)$ (pairwise) mutually unbiased bases (MUBs) can exist in a given dimension $D$. If $D$ is a prime power, e.g. $d=2^n$ ($n$ qubits), such maximal sets of mutually unbiased bases can be constructed \cite{Wootters_1989}. 
Viewed as a measurement primitive, such a collection of $D+1$ basis measurements is well conditioned and does allow for almost sample-optimal state tomography via projected least squares \cite{guta_fast_2018}.
Nonetheless, Hamiltonian Updates can struggle considerably with such a measurement primitive. 
The reason is that certain state pairs are extremely difficult to distinguish with MUB measurements \cite{Matthews_2009}.

As a concrete example, suppose that the (unknown) target state is pure and diagonal in the first MUB, say $\rho =|b_1 \rangle \! \langle b_1|$. In the first step of Algorithm~\ref{alg:motivation}, we need to be able to distinguish $\rho$ from the initial guess $\sigma_0 =\frac{1}{D}\mathbb{I}$. Mutual unbiasedness implies that $D$ out of the $D+1$ basis measurements fail to achieve this goal: $\left[p_U (\rho)\right]_i = |\langle c_i |b_1 \rangle|^2=\frac{1}{d}=\left[p_U (\sigma_0) \right]_i$ for all $i=1,\ldots,D$ and any basis $\left\{|c_i \rangle \right\}$ that is unbiased with respect to $\left\{|b_i \rangle \right\}$. In turn, a randomly selected MUB will produce a false positive with probability $D/(D+1)$. 
Hence, $L=\mathcal{O}(D)$ repetitions (inner loop) are required to obtain actionable advice in the first step of the algorithm alone! 
This number is already comparable to the total number of MUB settings and provides sufficient data for performing full quantum state tomography  (e.g.\ via projected least squares).

\subsection{Local measurements}
\label{sub:local-measurements}

The measurement primitives discussed in the previous subsection have one thing in common: they require circuits of moderate size -- $\mathcal{O}(n^2)$ for $n$ qudits -- to implement. Such \emph{global} unitary circuits are challenging to implement on current NISQ architectures \cite{Preskill2018}. 
In contrast, \emph{local} measurement primitives -- like performing independent single-qudit rotations, followed by computational basis measurements -- can routinely be carried out in various experimental platforms.
We refer to Fig.~\ref{fig:illustration} (right) for a visual illustration. 
The confined structure of local measurement primitives facilitates experimental implementation, but also renders a thorough theoretical analysis challenging. To this date, very few rigorous results address this setting and the achieved bounds on measurement settings and sample complexity are considerably worse than their more generic (global) counterparts, see Table~\ref{tab:local}.

Hamiltonian Updates can readily applied to this setting. 
What is more, distinguishability properties of $k$-local measurement primitives have already been studied in the literature. The main result in Ref.~\cite{Lancien_2013} states that the proportionality constant decays exponentially in the number $n/k$ of local constituents:
\begin{align}
 \mathbb{E}_{U \sim \text{$k$-loc.~4-design}} \left[ \|p_U (\rho) - p_U (\sigma) \|_{\ell_1} \right] 
\gtrsim &\left(\tfrac{1}{\sqrt{18}} \right)^{\tfrac{n}{k}} \Big(\sum_{I \subset \left\{1,\ldots,n/k\right\}}\| \mathrm{tr}_I (\rho-\sigma) \|_2^2 \Big)^{\tfrac{1}{2}}.
\label{eq:cecilia}
\end{align}
Here, $\mathrm{tr}_I (\cdot)$ denotes the partial trace over a collection of local constituents. 
While an exponential decay in $n/k$ is unavoidable in general, it is not known if the factor $1/\sqrt{18}$ captures the true decay. The norm $\|\rho-\sigma \|_{2(n/k)}=\big(\sum_{I \subset \left\{1,\ldots,n/k\right\}}\|\mathrm{tr}_I \left( \rho-\sigma \right) \|_2^2\big)^{1/2}$ also occurs naturally in the study of entanglement \cite{Harrow_2013,Lancien_2013}. It is always lower-bounded by $\|\rho-\sigma\|_2$, but can be considerably larger if the states in question are not too entangled. 
Unfortunately, translating the worst case interpretation $\mathbb{E}_{\text{$k$-local}} \left[ \|p_U (\rho) - p_U (\sigma) \|_{\ell_1} \right] \gtrsim (1/\sqrt{18})^{n/k} \|\rho-\sigma\|_2$ of Eq.~\eqref{eq:cecilia} into ensemble parameters \eqref{eq:tomo-sound-intro} does not produce competitive results for Hamiltonian Updates straight away. For qubit systems, local measurements addressing (at least) $k=5$ and $k=8$ qubits simultaneously are necessary to break even with PLS and CS in terms of measurement settings. This discrepancy becomes even more pronounced when considering sample complexity: local blocks of size (at least) $k=12$ and $k=21$ are required to reproduce the dimensional scaling of CS and PLS. 

Empirical studies conveyed by Fig.~\ref{fig:random+epr} in the main text suggest
that this is not a fundamental shortcoming of the proposed method, but a consequence of combining nontrivial, but probably still far from optimal, worst-case bounds. Indeed, consider the task of distinguishing a pure state $\ket{\psi}$ from the maximally mixed state, the first step of the algorithm when the target state is pure ($\rho = |\psi \rangle \! \langle \psi|$). For instance, suppose that $\ket{\psi}\sim\text{unif}$ is Haar random.
It is then not difficult to show that 
\begin{equation*}
   \mathbb{E}_{\ket{\psi}\sim\text{unif} } \mathbb{E}_{U \sim \text{$k$-loc.~4-design}} \left[ \|p_U (\ketbra{\psi}) - p_U (\mathbb{I}/D) \|_{\ell_1} \right] =\Omega(1).
\end{equation*}
That is, local measurements are (on average) as good as global ones in distinguishing pure states from the maximally mixed state. This average distinguishability bound is exponentially better than the worst-case bound in~\eqref{eq:cecilia}.
Thus, one of the main open questions left by this work is how to further reduce the gap between theory (rigorous bounds on sample complexity and runtime) and practice (numerical simulations and tractable implementation in the lab). The framework presented here leaves room for such improvements.
This, for instance, could entail proving an improved constant in Eq.~\eqref{eq:cecilia}, or a more thorough understanding of the distinguishability norm $\|\cdot \|_{2(n/k)}$. Either could then be converted into rigorous assertions about state tomography with local, single-shot measurements. 

\section{Stability with respect to state preparation and measurement errors}\label{sec:robust}

Quantum state tomography is an interesting theoretical problem in its own right, but should ultimately serve a practical purpose: help practitioners to properly scale up, calibrate and tune the quantum devices of today's NISQ era \cite{Preskill2018}. These devices are typically noisy and a practical tomography procedure should be capable of tolerating errors in both state preparation and measurement.
Perhaps surprisingly, existing competitive techniques struggle with this pre-requisite. Projected least-squares \cite{guta_fast_2018} is stable with respect to state preparation errors -- like drifting sources -- but
it is not known how miscalibration errors in the measurement affect the reconstruction quality. 
Compressed sensing techniques \cite{Gross2010,liu2011universal,kueng_low_2017} are even more fragile. The nontrivial reconstruction algorithm -- as well as the theoretical proof techniques required to provide rigorous convergence guarantees -- seem to be ill-equipped to handle even small measurement errors, see e.g.\ \cite{Roth2018,Roth2020} for a discussion and partial progress. 
In contrast, approaches based on mirror descent \cite{Tsuda2005,Bubeck2015} -- like Hamiltonian Updates (Algorithm~\ref{alg:HUtomo_last}) -- are designed to tolerate small errors in each update. These can either stem from state preparation, calibration errors in the measurement or inaccurate executions in the classical postprocessing. The first two examples address the most prominent noise sources in actual experiments, while the latter will allow us to considerably improve runtime by carrying out expensive steps (most notably: matrix exponentiation) only approximately. 

In \textsc{Hamiltonian Updates} (Algorithm~\ref{alg:HUtomo_last}) measurement data, and by extension: errors, only affect the estimated outcome distribution of the target state. This, in turn, can affect (and, to some extent, corrupt) the update rule $H \mapsto H+\eta UPU^\dagger$ for the Hamiltonian. A large step size $\eta$ can amplify this effect, while a small step size diminishes it. 
This is already the key idea for establishing stability: choose the step size $\eta$ sufficiently small -- as we shall see $\eta = \epsilon/16$ suffices -- to mitigate the effects of noise and imperfect implementation. 

To demonstrate this, let us assume that there is a constant $\epsilon_{\textrm{SP}}$ such that we perform each measurement on a (possibly measurement setting dependent) state $\tilde{\rho}$ that satisfies:
\begin{align}\label{equ:noisestate}
    \|\rho-\tilde{\rho}\|_{\mathrm{tr}}\leq \epsilon_{\text{state}}.
\end{align}
Similarly, we assume that a given basis measurement is also only approximately accurate. More precisely, ideal measurement channel $\mathcal{M}(\rho)=\sum_{i=1}^D \langle i| U \rho U^\dagger |i \rangle |i \rangle \! \langle i|$ and actual measurement channel $\tilde{\mathcal{M}}_U (\rho) = \sum_i \mathrm{tr} \left( E_i \right) |i \rangle \! \langle i|$ should be close in a meaningful worst-case fashion (induced trace norm):
\begin{equation}
\| \tilde{\mathcal{M}}_U - \mathcal{M}_U \|_{\mathrm{tr} \to \mathrm{tr}} 
= \max_{\rho \text{ state}} \sum_i \left|   \langle i| U \rho U^\dagger |i \rangle - \mathrm{tr}(E_i \rho) \right| \leq \epsilon_{\text{measurement}}.
\label{equ:noisemeasurement}
\end{equation}
In this setting, it is of course not possible to obtain an estimate that is closer than $\epsilon_{\text{state}}+\epsilon_{\text{measurement}}$ in trace distance, as our measurements cannot distinguish states that are this close due to imperfect control of state preparation and measurements.
However, mild adjustments ensure that Algorithm~\ref{alg:HUtomo_last} still converges to the true target state, provided that these noise effects are not too large.

\begin{thm}[Robustness of algorithm]\label{thm:spamresil}
The assertions of Proposition~\ref{prop:convergence-accuracy} and, by extension, Theorem~\ref{thm:main-detail}, remain valid in the presence of bounded state preparation \eqref{equ:noisestate} and measurement \eqref{equ:noisemeasurement} errors, provided that the accuracy parameter $\epsilon$ in Algorithm~\ref{alg:HUtomo_last} obeys $\epsilon > \epsilon_{\text{state}}+\epsilon_{\text{measurement}}$ and the step size $\eta$ is adjusted to obey $2 \eta \leq \epsilon - \epsilon_{\text{state}}-\epsilon_{\text{measurement}}$.
\end{thm}

\begin{proof}
The driving force behind each update in Algorithm~\ref{alg:HUtomo_last} is a projector $UPU^\dagger$ that discriminates the ideal outcome distribution $\left[q_i\right] = \langle i| U \rho U^\dagger |i \rangle$ of the target state from the computed outcome distribution $\left[p_i \right] = \langle i| U \sigma U^\dagger |i \rangle$ of the current iterate $\sigma_t$:
\begin{equation}
\mathrm{tr} \left( U^\dagger PU (\rho - \sigma_t) \right) = \sum_{p_i > q_i} \left( \langle i| U \rho U^\dagger |i \rangle - \langle i| U \sigma_t U^\dagger |i \rangle \right) = \tfrac{1}{2}\sum_{i=1}^D \left|p_i - q_i \right|
\label{eq:ideal-conversion}
\end{equation}
The larger this discrepancy, the more progress the update can achieve. To ensure constant step-wise progress, we require $\sum_{i} |p_i - q_i| \geq \epsilon$ before making an update.
Errors in state preparation -- i.e. preparing $\tilde{\rho}$ instead of $\rho$  -- and subsequent measurement -- i.e.\  estimating $\tilde{p}_i = \mathrm{tr}(E_i \tilde{\rho})$ instead of $p_i = \langle i| U \tilde{\rho} U^\dagger |i \rangle$ -- 
can, in principle, thwart Relation~\ref{eq:ideal-conversion}.
On the other hand, it should not come as a surprise that Rel.~\eqref{eq:ideal-conversion} is somewhat stable with respect to such perturbations.
More precisely, suppose that the prepared state $\tilde{\rho}$ is sufficiently close to the true target, i.e. $\|\rho-\tilde{\rho}\|_{\mathrm{tr}} \leq \epsilon_{\text{state}}$, and the actual measurement procedure does not deviate too much from the ideal one: $\max_{\rho \text{ state}} \sum_i\left|\langle i| U \rho U^\dagger |i \rangle - \mathrm{tr}(E_i \rho ) \right| \leq \epsilon_{\text{measurement}}$.
Then, the projector $\tilde{P}=\sum_{\tilde{p}_i >q_i} |i \rangle \! \langle i|$ constructed from such inaccurate data obeys
\begin{equation}
\mathrm{tr} \left( U^\dagger \tilde{P} U (\rho-\sigma_t) \right) \geq \tfrac{1}{2}\big(\sum_{i=1}^D \left| \tilde{p}_i - q_i \right| - \epsilon_{\text{state}}-\epsilon_{\text{measurement}}\big). \label{eq:noisy-conversion}
\end{equation}
Thus, conditioning on $\sum_i |\tilde{p}_i - q_i| >\epsilon$ is still enough to make constant progress provided that $\epsilon > \epsilon_{\text{state}}+\epsilon_{\text{measurement}}$ and the step size $\eta$ is adjusted appropriately. The proof of Proposition~\ref{prop:convergence} requires 
\begin{equation*}
2\eta < \tfrac{1}{2}\left(\epsilon - \epsilon_{\text{state}}-\epsilon_{\text{measurement}}\right).
\end{equation*}
The instantiation of Theorem~\ref{thm:spamresil} lists sufficient conditions to ensure this relation.
Thus it suffices to establish Rel.~\eqref{eq:noisy-conversion}. 
Start by replacing $\rho$ with $\tilde{\rho}$ at the cost of subtracting $\tfrac{1}{2} \|\tilde{\rho}-\rho \|_{\mathrm{tr}}\leq \tfrac{1}{2}\epsilon_{\mathrm{state}}$:
\begin{align*}
\mathrm{tr} \left( U^\dagger \tilde{P} U (\rho-\sigma_t) \right)
\geq & \mathrm{tr} \left( U^\dagger \tilde{P} U (\tilde{\rho}-\sigma_t) \right) - \tfrac{1}{2} \|\tilde{\rho}-\rho \|_{\mathrm{tr}}
\geq \mathrm{tr} \left( U^\dagger \tilde{P} U (\tilde{\rho}-\sigma_t) \right) - \tfrac{1}{2}\epsilon_{\text{state}}.
\end{align*}
Here, we have once more used Helstrom's theorem.
Next, we use the assumption that actual and ideal measurement differ by at most $\epsilon_{\text{measurement}}$ to complete the conversion:
\begin{align*}
\mathrm{tr} \left( U \tilde{P} U^\dagger (\tilde{\rho}-\sigma_t) \right)
=& \sum_{\tilde{p}_i >q_i} \left( \langle i| U \tilde{\rho} U^\dagger |i \rangle - \langle i| U \sigma_t U^\dagger |i \rangle \right) 
\\&\geq  \sum_{\tilde{p}_i >q_i} \left( \mathrm{tr}(E_i \tilde{\rho}) - \langle i| U \sigma_t U^\dagger |i \rangle \right) - \sum_{\tilde{p}_i>q_i} \left( \mathrm{tr}(E_i \tilde{\rho}) - \langle i| U \tilde{\rho} U^\dagger |i \rangle \right) \\
&\geq  \sum_{\tilde{p}_i >q_i} ( \tilde{p}_i - q_i ) - \tfrac{1}{2}\max_{\tilde{\rho}\text{ state}}\sum_i \left| \mathrm{tr}(E_i \tilde{\rho})-\langle i| U \tilde{\rho} U^\dagger |i \rangle \right| 
\\&\geq \tfrac{1}{2} \sum_i |\tilde{p}_i - q_i | - \tfrac{1}{2}\epsilon_{\text{measurement}}.
\end{align*}
Here, we have used that a sub-selected sum of differences between two probability distributions $\left[p_i \right]$ and $\left[q_i \right]$ obeys $\sum_{i \in I}(p_i -q_i) \leq \tfrac{1}{2}\sum_i |p_i-q_i|$ with equality if and only if $I = \left\{i:p_i >q_i \right\}$.
\end{proof}

Finally, we point out that a similar argument implies that the update rule and, by extension, the entire algorithm still performs correctly in the presence of statistical fluctuations. 
It suffices to estimate the outcome distribution $\left[q_i \right] = \langle i| U \rho U^\dagger|i \rangle$ up to accuracy $\epsilon_{\text{statistical}} < \epsilon - \epsilon_{\text{state}}-\epsilon_{\text{measurement}}$.

\section{Classical postprocessing complexity}\label{sec:classicalpostprocessing}
Let us analyse the complexity of implementing the classical processing required for Algorithm~\ref{alg:HUtomo_last}. We will phrase all the results in terms of the number of required iterations, error parameter $\epsilon>0$ and the parameters $\theta_{\mathcal{E}},\tau_{\mathcal{E}}$ defined in Eq.~\eqref{equ:measurement_ensemble_par} for the underlying measurement ensembles. We refer the reader to Table~\ref{tab:comparison} for the resulting complexity for different measurement ensembles.

We will start with a naive implementation to highlight the required steps.
For each iteration, we need to compute the updated Gibbs state $\sigma_t$ given access to a Hamiltonian. 
The obvious way of doing this is by diagonalizing $H$, computing $\textrm{exp}(-H)$ and $\tr{\textrm{exp}(-H)}$. Diagonalizing takes time $\cO(D^3)$ and the two other tasks $\cO(D)$.
Given the current Gibbs states $\sigma_t$, we need to compute the statistics with respect to the new measurements in the different bases. Given the unitaries $U_i$, this can be done by comparing the diagonals of $U_i^\dagger \sigma U_i$ with $p_i$. Computing each of the matrices $U_i^\dagger \sigma U_i$ takes time $\cO(D^3)$ and comparing takes $\cO(D)$. We conclude that each iteration can be done in time $\cO(D^3 L)$, where $L$ is maximum number of new measurement settings per iteration. As we have at most $\cO(\log(D)\epsilon^{-2})$ iterations, the total runtime is at most $\cO(D^3 L\log(n)\epsilon^{-2})$. 

Although this runtime is already comparable or even faster than state-of-the-art~\cite{Gross2010,guta_fast_2018}, we will now discuss how to further exploit the structure and freedom of the algorithm to obtain a $\cO(D^2)$ runtime.
The following property will be key for this:
\begin{defi}[Fast matrix-vector multiplication property]
A measurement ensemble $\mathcal{E}$ over the unitary group of dimension $D$ is said to have the fast matrix-vector multiplication property (FMVM) if for all $U$ in of the ensemble we have that matrix-vector multiplication by $U$ and $U^\dagger$ can be done in $\tcO(D)$ time.
\end{defi}

As we will show later, many different choices of measurement ensembles enjoy this property. Examples include random Cliffords and various approximate $t$ design constructions in the literature. We refer the reader to Appendix~\ref{sec:fastvecmulti} for a proof of this fact and more details on this.

We will now see how to exploit the fact that we can perform vector-matrix multiplication faster to speedup the implementation of our algorithm. The next lemma will be crucial for that:

\begin{lem}\label{lem:errortruncation}
Fix a Hermitian $D \times D$  matrix $H$, an accuracy $\epsilon$ and let $l$ be the smallest even number that obeys $(l+1)(\log (l+1)-1) \geq 2\| H \|+ \log (D)+ \log (1/\epsilon)$. Then, the truncated matrix exponential $T_l = \sum_{k=0}^l \tfrac{1}{k!} (-H)^k$ is guaranteed to obey
\begin{equation*}
\left\| \frac{\exp (-H)}{\tr{\exp (-H)}} - \frac{T_l}{ \tr{T_l}} \right\|_{tr} \leq \epsilon.
\end{equation*}
Moreover, $\frac{T_l}{ \tr{T_l}}$ is a quantum state.
\end{lem}

\begin{proof}
We refer to~\cite[Lemma 3.2]{Brandao2019sdp} for a proof.
\end{proof}

As we saw before in Theorem~\ref{thm:spamresil}, it suffices to obtain $\epsilon/8$ approximations in trace distance to $\sigma_t$ at each iteration to run our algorithm. Thus, the lemma above allows us to work with the truncated Taylor series instead of the actual Gibbs state, which leads to significant speedups.

\begin{lem}\label{lem:computemi}
Let $\sigma_H=\textrm{exp}(-H)/\tr{\textrm{exp}(-H)}$ be the Gibbs state of one of the iterations of Algorithm~\ref{alg:HUtomo_last} and suppose that the measurement ensemble $\mathcal{E}$ has the FMVM property. 
Then we can compute $M_{U_i}(\sigma)$ up to an error $\cO(\epsilon)$ in total variation distance in time $\tcO(D^2m\epsilon^{-1})$.
\end{lem}
\begin{proof}
First note that the Hamiltonian $H$ has the form:
\begin{align*}
    H=\sum\limits_{i=1}^m U_iD_iU_i^\dagger,
\end{align*}
where $U_i$ were drawn from $\mathcal{E}$ and $D_i$ is a diagonal matrix. Now, note that at each iteration of the algorithm~\ref{alg:HUtomo_last} we increase the norm of the Hamiltonian by at most $\cO(\epsilon)$, as we add a term with operator norm $\cO(\epsilon)$. As there are at most $\cO(\log(D)\epsilon^{-2})$ iterations, we see that:
\begin{align*}
    \|H\|=\cO(\log(D)\epsilon^{-1}).
\end{align*}
Lemma~\ref{lem:errortruncation} implies that picking $l=\cO(\log(D)\epsilon^{-1})$ is enough to ensure that $T_l/\tr{T_l}$ will be $\epsilon$ close in trace distance to $\sigma$, where again:
\begin{align*}
   T_l= \sum_{k=0}^l \tfrac{1}{k!} (-H)^k.
\end{align*}
Let us now discuss how to compute $\tr{T_l}$. This is, of course, equivalent to computing $\bra{i}T_l\ket{i}$ for all different computational basis elements.
As we assumed that we have the FMVM property, it follows that we can compute $U_i D_i U_i^\dagger\ket{i}$ in time $\tcO(D)$, as $D_i$ is diagonal and, thus, we can perform matrix vector multiplication in time $\tcO(D)$ for $D_i$ and $U_i,U_i^\dagger$. This implies that we can compute $H\ket{j}$ in time $\tcO(Dm)$ by computing each term individually and summing up the corresponding vectors. Moreover, we can apply the same procedure to the resulting vector $H\ket{j}$ and compute $H^2\ket{j}$ in the same time. Iterating this argument, we see that we can compute $H^k\ket{j}$ in time $\tcO(kDm)$. Thus, we can compute $\bra{i}T_l\ket{j}$ in time $\tcO(Dml\log(D))$ and compute the trace in time $\tcO(D^2ml)$.
Furthermore, $\tr{U_i\ketbra{j}U_i^\dagger T_l}$ can be computed in exactly the same way as we computed the trace, but now note that the starting vector is $U_i\ket{j}$. We conclude that it takes time:
\begin{align*}
    \cO(D^2ml\log(D))=\tcO(D^2m\epsilon^{-1})
\end{align*}
to compute $M_{U_i}(\frac{T_l}{\tr{T_l}})$. As $\frac{T_l}{\tr{T_l}}$ is $\epsilon$ close in trace distance to $\sigma_t$, we have: 
\begin{align*}
    \left\|M_{U_i}\lb \sigma-\frac{T_l}{\tr{T_l}}\rb\right\|_{\ell_1}\leq\left\|\sigma-\frac{T_l}{\tr{T_l}}\right\|_{tr}\leq\epsilon,
\end{align*}
which yields the claim.
\end{proof}
Let us now discuss the complexity of outputting a Hamiltonian that describes a Gibbs state that is $\epsilon$ close to the target state in trace distance.
\begin{cor}[Complexity of classical postprocessing]\label{cor:timecomplexity}
Let $\rho\in \M_n$ be a state of rank at most $r$.
Then Algorithm~\ref{alg:HUtomo_last} can be run with a FMVM ensemble $\mathcal{E}$ and the same parameters and recovery guarantees as in Theorem~\ref{thm:main-informal} in time at most:
\begin{align*}
    \tcO(D^2r^{\frac{5}{2}}\tau_{\mathcal{E}}(\rho)^{-1}\theta_{\mathcal{E}}(\rho)^{-5}\log(\delta^{-1})\epsilon^{-5}).
\end{align*}

\end{cor}
\begin{proof}
Let us break down the steps of the algorithm and the corresponding costs.
At each iteration $t$ we must compute $M_U(\sigma_t)$ up to precision $\cO(\epsilon \theta_{\mathcal{E}}(\rho)^{-1}r^{-\frac{1}{2}})$ for at most $\cO(\tau_{\mathcal{E}}(\rho)^{-1}\log(\delta^{-1}))$ different unitaries, where $\sigma_t$ is the state at iteration $t$. 
It follows from Lemma~\ref{lem:computemi} that this task takes $\tcO(D^2m\theta_{\mathcal{E}}(\rho)^{-1}\epsilon^{-1})$.
Thus, we conclude that the total cost per iteration is
\begin{align}\label{equ:costofiterations}
    \tcO(D^2m\theta_{\mathcal{E}}(\rho)^{-1}r^{\frac{1}{2}}\epsilon^{-1}\tau_{\mathcal{E}}(\rho)^{-1}\log(\delta^{-1})).
\end{align}
Moreover, the number of iterations is at most
\begin{align}\label{equ:numberofiterations}
    m=\tcO(r\left[\theta_{\mathcal{E}}(\rho)\epsilon\right]^{-2}).
\end{align}
Multiplying~\eqref{equ:costofiterations} by $m$ gives the total cost of the algorithm, and inserting the bound on $m$ given in~\eqref{equ:numberofiterations} yields the claim.
\end{proof}
To the best of our knowledge, this algorithm outperforms all existing rigorous tomography algorithms in its scaling in the regime where $\epsilon^{-5}\tau_{\mathcal{E}}(\rho)^{-1}\theta_{\mathcal{E}}(\rho)^{-5}=o(D)$. This is e.g. the case for random approximate $4$ designs or Cliffords and $\epsilon$ constant, as we summarize in more detail in Table~\ref{tab:comparison}.
The best available algorithms~\cite{guta_fast_2018} in terms of computational complexity of the postprocessing scale at least like $\cO(D^3)$, as they require at least one diagonalization of a matrix.
However, the worst-case dependency in $\epsilon$ can be $\epsilon^{-5}$ and it would be interesting to try to improve this dependency.

\subsection{Memory requirements, parallelization and other features}\label{sec:nicefeatures}
Another attractive feature of our algorithm is that it can be run with almost optimal memory requirements, i.e., we essentially only need to store classical descriptions of the underlying unitaries and the observed statistics in that basis. 
More precisely:
\begin{thm}[Memory requirements]\label{thm:classmem}
Let $\mathcal{S}_{\mathcal{E}}$ be the maximum memory required to store a unitary from an ensemble $\mathcal{E}$.
Then Algorithm~\ref{alg:HUtomo_last} can be run with error parameter $\epsilon$ and failure probability at most $1-\delta$ requiring at most
\begin{align*}
    \cO((D+\mathcal{S}_{\mathcal{E}})T \log(D\epsilon^{-1})),
\end{align*}
classical memory, where $T=\lceil 32 \log (D) r/\theta_{\mathcal{E}}^2(\rho)\epsilon^{-2} \rceil$ is the maximum number of iterations.
\end{thm}
\begin{proof}
In order to store the Hamiltonian at each iteration, we need to store the at most $T$ different probability distributions $p_i$ for the different measurement outcomes, the diagonal matrices in the description of the Hamiltonian and the corresponding unitaries. 
It suffices to store each entry of the vectors $p_i$ and diagonals up to a precision $\epsilon^2 D^{-1}$ for our purposes, as this is the precision we have for the empirical distribution. Storing the  vectors $p_i$ and the diagonal matrices up to a precision $\epsilon^2 D^{-1}$ for each entry takes $\cO(DT\log(D\epsilon^{-1}))$ classical memory. Storing the unitaries takes up $\cO(T\mathcal{S}_{\mathcal{E}})$ memory.
Let us now discuss the memory requirements for running the algorithm. At each iteration, we need to approximately compute $M_i(\sigma_t)$ for the current guess and compare it to $p_i$. We will follow the strategy devised in Lemma~\ref{lem:computemi} to approximately compute $M_i(\sigma_t)$.
Thus, it suffices to compute $\bra{j}U_i^\dagger T_lU_i\ket{j}$ for each $j$ separately, where $T_l$ is defined as in Lemma~\ref{lem:errortruncation}. Let us now discuss how to compute $\bra{j}U_i^\dagger T_lU_i\ket{j}$ only using $\tcO(Dm)$ classical memory. We will compute $T_lU_i\ket{j}$ recursively. Set $x_k^{i,j}=H^kU_i\ket{j}$. We clearly have:
\begin{align*}
    \bra{j}U_i^\dagger T_lU_i\ket{j}=\sum\limits_{k=0}^{l}\frac{1}{k!}\bra{j}U_i^\dagger\ket{x_k^{i,j}}
\end{align*}
Let $H$ be the Hamiltonian at time $t$ and $H=\sum\limits_{v=0}^t U_vD_vU_v^\dagger$.
For computing $HU_i\ket{j}$, we write:
\begin{align*}
x_1^{i,j}=HU_i\ket{j}=\sum\limits_{v=1}^t U_vD_vU_v^\dagger U_i\ket{j}
\end{align*}
and compute each term separately. This takes $\cO(mD)$ classical memory. We then store $x_1^{i,j}$ in the memory and compute $Hx_1^{i,j}$ in analogous manner. We then store $x_2^{i,j}=Hx_1^{i,j}$ and add $x_2^{i,j}/2$ to $x_0^{i,j}+x_1^{i,j}$, while deleting $x_1^{i,j}$. Repeating this procedure until $x_l^{i,j}$ we see that we can compute $U_i^\dagger T_lU_i\ket{j}$ using at most $\tcO(nD)$ classical memory. We then compute $\bra{j}U_i^\dagger T_lU_i\ket{j}$ and store the corresponding value. We conclude that we can compute $M_i(\sigma)$ approximately using at most $\tcO(DT)$ classical memory. It follows that we can store all the relevant information for running an iteration and obtain and store all information required to check a violation of the constraints again with $\tcO(DT)$ memory. This concludes the proof.
\end{proof}
We note that for certain measurement setups, such as approximate $4-$designs given by random local quantum circuits, the resulting required memory for doing tomography of a rank $r$ quantum state will be $\tcO(Dr)$, which is optimal up to logarithmic factors.

Furthermore, our algorithm can be parallelized easily. That is, it is possible to compute $\bra{j}U_i^\dagger T_lU_i\ket{j}$ for different $j$ on parallel processors, as we only need to provide them with a description of $H$. Moreover, at each iteration, updating the description of $H$ takes $\cO((D+\mathcal{S}_{\mathcal{E}}) \log(D\epsilon^{-1}))$ time.

Another feature which is relevant for large-scale applications is the online flavour of the algorithm. That is, it is possible to already start the classical postprocessing procedure as data is acquired and it is straightforward to add new measurement results to the algorithm, which can also significantly shorten the overall time required to do tomography.

Furthermore, we mention in passing is that the algorithm can naturally incorporate further prior information on $\rho$. More precisely, if we know that $\rl{\rho}{\rho'}$ is small for some already known state $\rho'$, then we can use $\rho'$ as the starting state for the algorithm and obtain a faster convergence.

Thus, the essentially optimal memory requirements of our algorithm and straightforward parallelization renders it practical for large scale applications and give it a significant advantage over all existing tomography procedures. We summarize the exact scaling of the memory requirements for different measurement setups in Table~\ref{tab:comparison}.

\section{Implementation on a quantum computer}\label{sec:quantumimplementation}
Our algorithm allows for a straightforward implementation on a quantum computer. Note that all that our algorithm requires at each iteration are the statistics of the quantum state in different bases. 
By preparing $\tcO(D\epsilon^{-2})$ copies of the Gibbs state on a quantum computer and measuring it in the basis specified by a unitary then suffices to obtain the statistics w.r.t. to that basis up to an error $\epsilon$ in $\ell_1$ distance. Thus, it is possible to perform each iteration of our tomography algorithm by preparing enough copies of the underlying Gibbs state for the current guess.

We will only discuss the implementation of the algorithm for measurements ensembles given by a random approximate $4$ design that can be implemented in $\tcO(1)$ time, but it should be straightforward to adapt the results to other measurements.

 There are many different proposals for preparing Gibbs states on 
 quantum computers~\cite{Chowdhury2016,Franc2017,Kastoryano2014,Poulin2009,Temme2009,Temme2009,Yung2012,Apeldoorn2017}.
 Here, we will follow the algorithm proposed in~\cite{Poulin2009}. Their results reduce the problem of preparing $\rho_H = \exp (-H)/\tr{\exp (-H)}$
to the task of simulating the Hamiltonian $H$.
Indeed,~\cite{Poulin2009} shows that $ \tcO\lb\sqrt{\frac{D}{Z_H}}\epsilon^{-3}\rb$ queries to the entries of a controlled $U$, where $Z_H=\tr{e^{-H}}$ and $U$ satisfies
 \begin{align*}
  \|U-e^{it_0H}\|\leq\cO(\epsilon^3) \quad \textrm{where} \quad t_0 = \pi / (4\|H \|)
 \end{align*}
suffice to produce a state that is $\epsilon$ close in trace distance to $\rho_H$. The probability of failure is at most $D^{-1/\epsilon}\epsilon^{2}$.   
By construction, the Hamiltonians we wish to simulate are all of the form
\begin{align*}
    H=\sum\limits_{i=1}^mU_iD_iU_i^\dagger,
\end{align*}
where $D_i$ are diagonal matrices and $U_i$ can be implemented in $\tcO(1)$ time.
It follows from~\cite[Theorem 1]{Childs2012} that  
\begin{align*}
\tilde{\cO}\lb t_0m^2 \exp (1.6\sqrt{\log m^2\log(n)t\epsilon^{-1}})\rb  
\end{align*}
separate simulations of $U_ie^{it_0D_i}U_i^\dagger$ suffice to simulate $H$ for time $t_0$ up to an error $\epsilon$. Thus, we further reduce the problem of simulating $H$ to simulating the $U_ie^{it_0D_i}U_i^\dagger$. As, by assumption, we can generate $U_i$ in $\tcO(1)$ time, we will focus on simulating the diagonal Hamiltonians $D_i$.
Let $O_{D_i}$ be the matrix entry oracle for $D_i$. We suppose that it acts on $\setC^D\otimes\lb\setC^2\rb^{\otimes l}$, where $l$ is large enough to represent the diagonal
entries to desired precision in binary, as 
\begin{align}
O_{D_i}\ket{j,z}\mapsto\ket{j,z\oplus D_{jj}}. 
\end{align}
It is then  possible to simulate $D_i$ for times $t=\tilde{\cO}(\epsilon^{-1})$
with $\tilde{\cO}(1)$ queries to the oracle $O_{D_i}$ and elementary operations~\cite{Berry2007}.
Thus, efficient simulation of $\mathrm{e}^{-iD_it}$ follows from an efficient implementation of the oracle $O_{D_i}$.
The latter can be achieved with a quantum RAM~\cite{Giovannetti2007}. 
We consider the quantum RAM model from \cite{Prakash2014}.
There, it is
possible to make insertions in time $\tcO\lb1\rb$. 
Thus, given a classical description of a diagonal matrix $D_i$, we may
update the quantum RAM in time $\tcO\lb D\rb$. After we have updated the quantum RAM, we may implement the oracle $O_{D_i}$ in time $\tcO(1)$.
Combining all these subroutines establishes the quantum runtime of our algorithm:
\begin{prop}
Let $\rho$ be a quantum state of rank at most $r$. Then we can run
Algorithm~\ref{alg:HUtomo_last} with the parameters and recovery guarantees specified in Thm.~\ref{thm:main-detail}  in time $\tcO(D^{3/2}r^3\epsilon^{-9})$ on a quantum computer.
\end{prop}
\begin{proof}
Let $\sigma_t$ be the current guess for the state.
Note that it requires $\cO(Dr\epsilon^{-2})$ copies of $\sigma_t$ to estimate the statistics w.r.t. to a given basis up to an error $\epsilon /\sqrt{r} $ in total variation distance, the precision required by Thm.~\ref{thm:main-detail}. For each iteration, we will have to check at most $\tilde{\cO}(1)$ different bases. 
Given the different measurements statistics $\left[q_i \right] = \langle i| U \sigma_t U^\dagger |i \rangle$, we can check for violations in $\cO(D)$ (classical) time.
Thus, the complexity of each iteration is dominated by the cost of  preparing the $\tilde{\cO}(Dr\epsilon^{-2})$ copies of $\sigma_t$ -- a Gibbs state.
Moreover, we will have at most 
$
T=\tcO(r\epsilon^{-2})
$
updates before reaching convergence (Proposition~\ref{prop:convergence}).
Thus, the entire execution of the algorithm requires at most
$
\tcO(Dr\epsilon^{-2}T)=\tcO(Dr^2\epsilon^{-4})
$ Gibbs state preparations.
According to the discussion above, each Gibbs state can be generated in time
$
    \tcO(T^2\sqrt{D}\epsilon^{-3})
$.
This results in a total (quantum) runtime of order
\begin{align*}
    \tcO((T^2\sqrt{D}\epsilon^{-3}\times Dr^2\epsilon^{-4})=\tcO(D^{3/2}r^3\epsilon^{-9}).
\end{align*}
\end{proof}
Note that the quantum algorithm also outputs a classical description of the quantum state in terms of the
(diagonal) projectors $P_i$ and associated basis changes $U_i$.
To the best of our knowledge, this is the first quantum speedup for tomography beyond the results of~\cite{kerenidis2018quantum}. 

There, the authors show how to do tomography for a \emph{real, pure} state $\ket{\psi}$ up to an error $\epsilon$ in trace distance given access to a \emph{controlled unitary} preparing copies of $\ket{\psi}$ only using $\tcO(D\epsilon^{-2})$ copies of $\ket{\psi}$ and $\tcO(D\epsilon^{-2})$ classical postprocessing. They also assume access to a QRAM.
However, this remarkable result
addresses a very different setup.
Although the authors comment that it is possible to adapt their results to go beyond states with only real phases, it is unclear how to extend it to states that are not (exactly) pure. More importantly, the protocol is contingent on  the assumption that one is able to produce copies of the target state with a controlled unitary -- a manifestly stronger state preparation model than the i.i.d. setting discussed here.

Finally, we point out that the scaling in terms of accuracy is considerably worse: $\epsilon^{-9}$ for the quantum implementation vs.\ $\epsilon^{-5}$ for the classical one. We leave a reduction of this gap to future work.

\section{Efficiently computing approximate eigenvectors and eigenvalues of the target state}\label{sec:convertingfromgibbstousual}

Efficient implementations of Hamiltonian Updates (Algorithm~\ref{alg:HUtomo_last}) do not output the estimated state $\sigma_\star$ itself, but a Hamiltonian $H_\star$ that fully characterizes the solution: $\sigma_\star = \exp (-H_\star)/\mathrm{tr}(\exp (-H_\star))$.
Although this provides a complete description of the state, it might be desirable for some applications to output the state in a more traditional form, i.e. in terms of a list of eigenvalues and corresponding eigenvectors. Let us now discuss how we can convert the output of our algorithm to this more traditional representation efficiently.
Assuming that the target $\rho$ has rank $r$, we can conclude that the algorithm also outputs a Gibbs state $\sigma_\star$ that is well-approximated by a rank-$r$ density matrix.
We can once again capitalize on fast matrix-vector multiplication with $H$
to obtain approximate eigenvectors and eigenvalues in $\tcO(Dr)$ time instead of the usual $\cO(D^3)$, while also using only $\tcO(Dr)$ memory. We start by recalling the following result of~\cite[Corollary 4.4]{Steurer2015}:
\begin{lem}
Set $
    S_l=\sum\limits_{k=0}^l\frac{(-H)^{k}}{k!}
$ with $l\leq 3e(\|H\|+\log(\epsilon^{-1}))$. Then,
\begin{align*}
    \left\|\frac{e^{-H}}{\mathrm{tr}(e^{-H})}-\frac{S_l^2}{\tr{S_l^2}}\right\|_{\mathrm{tr}}\leq\epsilon.
\end{align*}
\end{lem}

The proof follows from a Taylor expansion argument, see \cite[Corollary 4.4]{Steurer2015} for more details.

The state $A=\frac{S_l^2}{\tr{S_l^2}}$ has much in common with $T_l$, but the main difference we will exploit  is that it is simple to compute its square root. This property will turn out to be key in the analysis that follows. We will exploit the main result of~\cite{musco2015randomized} to obtain an approximate list of eigenvalues and eigenstates. Their main algorithm is described in Algorithm~\ref{alg:krylov} and its output $Z$ can be used to obtain a good low-rank approximation of $\sqrt{A}$, as we will see in Theorem~\ref{thm:fromgibbstousual}. We will then combine this with the gentle measurement Lemma~\cite{Winter_1999} to obtain a good approximation of our state.
\begin{algorithm}[H]
\caption{\textit{Block Krylov Iteration}
}
\label{alg:krylov}
\begin{algorithmic}[1]
\Require{maximal rank $r$, error $\epsilon$ and $A=S_l/\sqrt{\tr{S_l^2}}$.}
\State{Set $q=\cO(\log(D)\epsilon^{-1/2})$ and draw $X\sim\mathcal{N}(0,1)^{D\times r}$}
\State{Compute $K=\left[AX,A^3X,\ldots,A^{2q+1}X\right]$}
\State{Orthonormalize the columns of $K$ to obtain $Q$.}
\State{Compute $Y=Q^\dagger A^2Q$}
\State{Set $U_r$ to be the top $r$ singular vectors of $Y$}
\State{Return $Z=QU_r$}
\end{algorithmic}
\end{algorithm}

We will first show that we can run Algorithm~\ref{alg:krylov} efficiently.
\begin{lem}
    Let $H$ be the output of Algorithm~\ref{alg:HUtomo_last} with a measurement ensemble $\mathcal{E}$ with the FVMM property. Then we can run Algorithm~\ref{alg:krylov} in time at most $\tcO(DTr^2\epsilon^{-\frac{3}{2}})$, where $T=\lceil 32 \log (D) r/\theta_{\mathcal{E}}^2(\rho)\epsilon^{-2} \rceil$ is the maximal number of iterations.
\end{lem}
\begin{proof}
Recall that we can multiply a vector with $H$ in time $\tcO(DT)$. Thus, we can also multiply a vector with $S_l$ in time $\tcO(DTl)=\tcO(DT\epsilon^{-1})$ and computing $K$ takes time $\tcO(DTrq\epsilon^{-1})=\tcO(DTr\epsilon^{-3/2})$, as $X$ has $r$ columns.
Orthonormalizing $K$ takes time $\tcO(Dqr)$ and computing $Y$ again takes time $\tcO(DTr\epsilon^{-1})$, as $Q$ has $rq$ columns. Computing the SVD of $Y$ then takes time $\tcO(r^3)$, as it is a $qr\times qr$ matrix. Finally, multiplying $Q$ by $U_r$ takes time $\tcO(r^2D)$. We see that all steps are individually bounded in runtime by $\tcO(DTr^2\epsilon^{-\frac{3}{2}})$ and the claim follows.
\end{proof}
We assumed we know $\tr{S_l^2}$ in the definition of the algorithm, but note that as it is a positive constant, we can also run Algorithm~\ref{alg:krylov} with $S_l$ instead and the output $P$ will be the same.
We are now ready to show that $Z$ can be used to obtain a good approximation to $\rho$:
\begin{thm}\label{thm:fromgibbstousual}
Let $Z$ be the output of Algorithm~\ref{alg:krylov} and set $P=Z^\dagger Z$ (an orthoprojector). 
Suppose that the output $H_\star$ of Algorithm~\ref{alg:HUtomo_last} satisfies
\begin{align*}
    \left\|\rho-\frac{e^{-H_\star}}{\tr{e^{-H_\star}}}\right\|_{\mathrm{tr}} = \| \rho - \sigma_\star \|_{\mathrm{tr}}\leq \epsilon,
\end{align*}
for a target state $
\rho$ with rank at most $r$. Then 
\begin{align}
    \left\|\rho-\frac{PS_l^2P}{\tr{PS_l^2P}}\right\|_{tr}=\cO(\sqrt{\epsilon})
\end{align}
with high probability.
\end{thm}
\begin{proof}
In~\cite[Theorem 1]{musco2015randomized}, the authors show that $P$ satisfies
\begin{align}\label{equ:initialestimate}
    \left\|\frac{S_l}{\sqrt{\tr{S_l^2}}}-P\frac{S_l}{\sqrt{\tr{S_l^2}}}\right\|_2\leq (1+\epsilon)\left\|\frac{S_l}{\sqrt{\tr{S_l^2}}}-A_r\right\|_2,
\end{align}
with high probability, where $A_r$ is an arbitrary matrix of rank $r$.
Let us now estimate this distance when we pick $A_r=\sqrt{\rho}$. As both $\frac{S_l}{\sqrt{\tr{S_l^2}}}$ and $\sqrt{\rho}$ are square roots of states:
\begin{align*}
   \left\|\frac{S_l}{\sqrt{\tr{S_l^2}}}-\sqrt{\rho}\right\|_2^2=2-\tr{\frac{S_l}{\sqrt{\tr{S_l^2}}}\sqrt{\rho}}\leq\left \|\rho-\frac{S_l^2}{\tr{S_l^2}}\right\|_{tr}.
\end{align*}
See e.g.~\cite[Eq. 3]{audenaert2012comparisons} for the last inequality. 
It then follows by combining our assumption that we have a good approximation in trace distance to $\rho$ from $H_*$ and the fact that $\frac{S_l^2}{\tr{S_l^2}}$ approximates the Gibbs state that, combined with a triangle inequality:
\begin{align*}
     \left\|\rho-\frac{S_l^2}{\tr{S_l^2}}\right\|_{tr}\leq \left\|\rho-\frac{e^{-H}}{\tr{e^{-H}}}\right\|_{tr}+ \left\|\frac{e^{-H}}{\tr{e^{-H}}}-\frac{S_l^2}{\tr{S_l^2}}\right\|_{tr}\leq2\epsilon.
\end{align*}
This, combined with Eq.~\eqref{equ:initialestimate}, yields that:
\begin{align*}
     \left\|\frac{S_l}{\sqrt{\tr{S_l^2}}}-P\frac{S_l}{\sqrt{\tr{S_l^2}}}\right\|_2= \cO(\sqrt{\epsilon}).
\end{align*}
We now have:
\begin{align*}
     \left\|\frac{S_l}{\sqrt{\tr{S_l^2}}}-P\frac{S_l}{\sqrt{\tr{S_l^2}}}\right\|_2^2=1-\tr{\frac{PS_l^2P}{\tr{S_l^2}}}= \cO(\epsilon).
\end{align*}
As $P$ is a projection, we have that the probability we observe the outcome $P$ when measuring it on $\frac{S_l^2}{\tr{S_l^2}}$ is at least:
\begin{align*}
    \tr{\frac{PS_l^2P}{\tr{S_l^2}}}\geq1- \Omega(\epsilon)
\end{align*}
Thus, by the gentle measurement Lemma~\cite{Winter_1999}:
\begin{align*}
    \left\|\frac{PS_l^2P}{\tr{PS_l^2P}}-\frac{S_l^2}{\tr{S_l^2}}\right\|_{tr}=\cO(\sqrt{\epsilon}).
\end{align*}
A series of triangle inequalities gives the claim.
\end{proof}

As we can compute $PS_l^2P$ in time $\tcO(DTr\epsilon^{-1})$, we conclude that it is possible to convert the output of Algorithm~\ref{alg:HUtomo_last}, given as a Hamiltonian $H_\star$, into a more traditional form: a collection of eigenvectors and eigenvalues.
The runtime for this postprocessing step is comparable to the time required to find $H_\star$.

To see this, set 
$
  P=\sum\limits_{i=1}^r\ketbra{\psi_i}
$, where the $\ket{\psi_i}$ correspond to the rows of the matrix $Z$ we output in Algorithm~\ref{alg:krylov}.
We can then compute the $r\times r$ matrix
\begin{align*}
    B_{i,j}=\frac{\bra{\psi_j}S_l^2\ket{\psi_i}}{\tr{PS_l^2P}}
\end{align*}
in time $\tcO(D r^2)$ by computing $S_l\ket{\psi_i}$ and the corresponding scalar products. By diagonalizing this $r\times r$ matrix $B$, we can recover eigenvalues and present the eigenvectors as linear combinations of the $\ket{\psi_i}$'s.
Furthermore, note that the algorithm presented here only requires classical memory of size  $\cO(DTr)$.

The only relevant property of $\rho$ we used for the proof above is that it is of low-rank and we are able to multiply fast with $S_l$. Thus, if we know that the current iteration of our algorithm is already close to a low-rank state, it is possible to use the algorithm above to reduce the complexity of performing an eigenvalue decomposition.

\section{Effective rank} \label{sec:effective-rank}
Here we collect some statements about the effective rank relevant to our work.
We define the $\alpha$-effective rank of a quantum state $\rho$ as
\begin{equation}
r_{\mathrm{eff},\alpha}(\rho)= \tr{\rho^{\alpha}}^{\frac{1}{1-\alpha}} \quad \text{for $\alpha \in (0,1)$.}
\end{equation}
This is just the exponential of the $\alpha$-R\'enyi entropy of the quantum state $\rho$, defined as
\begin{align*}
    S_\alpha(\rho)=\frac{1}{1-\alpha}\log\lb \tr{\rho^{\alpha}}\rb=\log\lb r_{\mathrm{eff},\alpha}(\rho)\rb.
\end{align*}
It is well-known that this quantity is monotonically decreasing in $\alpha$ with limit
\begin{align*}
\lim_{\alpha\to 0}    S_\alpha(\rho)=\log(\text{rank}(\rho)).
\end{align*}
This, in particular, implies
$r_{\mathrm{eff},\alpha}(\rho) \leq \mathrm{rank}(\rho)$ for every $\alpha \in (0,1)$. Moreover, $r_{\mathrm{eff},\alpha}(\rho)$ is a continuous function of the state. 

On a more conceptual level, these functions are known to capture how fast the spectrum of $\rho$ decays~\cite{PhysRevB.73.094423}.
More precisely, let $\rho$ be a state with eigenvalues
$\lambda_1,\ldots,\lambda_D$ arranged in non-increasing order. For $1 \leq r \leq D$ (integer) define
\begin{align*}
    \tau(r,\rho)=\sum\limits_{k=r+1}^D\lambda_k.
\end{align*}
This quantity -- the sum of the $D-r$ smallest eigenvalues (``tail'') -- captures how well a quantum state is approximated by a rank 
$r$ state. The $\alpha$-entropies control how fast this tail decays.
A majorization argument shows that
\begin{align}\label{equ:taildecay}
    \tau \lb r,\rho\rb\leq\lb 1-\alpha\rb^\frac{1}{\alpha} \lb\frac{r_{\mathrm{eff},\alpha}(\rho)}{r}\rb^{\frac{1-\alpha}{\alpha}},
\end{align}
see e.g.~\cite[Lemma 2]{PhysRevB.73.094423}. All of these properties justify the choice of $r_{\mathrm{eff},\alpha}$ as a continuous relaxation of the rank.

Let us now show an equivalence inequality between the trace norm and Frobenius norm tailored to low-rank states. We will then later generalize it to states of small effective rank.

\begin{lem}\label{lem:stablerankfrobenius}
For two quantum states $\rho,\sigma$ we have:
\begin{equation*}
   \| \rho - \sigma \|_1 \leq 2 \sqrt{\min \left\{ \mathrm{rank}(\rho),\mathrm{rank}(\sigma) \right\}} \| \rho - \sigma \|_2. 
\end{equation*}
\end{lem}

\begin{proof}
Helstrom's theorem connects the trace distance of $\rho$ and $\sigma$ with optimal distinguishing measurements:
\begin{equation*}
\| \rho - \sigma \|_1 = 2 \max_{0 \leq M \leq \mathbb{I}} \mathrm{tr}(M (\rho - \sigma) ).
\end{equation*}
Equality occurs if and only if $M$ is the orthoprojector onto the positive range of $\rho - \sigma$: $M^\sharp=P_+$ (or its ortho-complement, the projector onto the negative rank).
By construction, the range of $P_+$ is contained in the range of $\rho$ and we conclude $\|P_+ \|_2=\sqrt{\mathrm{tr}(P_+)} \leq \sqrt{\mathrm{rank}(\rho)}$. Combine this insight with Cauchy-Schwarz to obtain
\begin{equation*}
\| \rho - \sigma \|_1 = 2\mathrm{tr} \left( P_+ (\rho-\sigma) \right) \leq 2 \| P_+ \|_2 \| \rho - \sigma \|_2
\leq 2 \sqrt{\mathrm{rank}(\rho)} \| \rho - \sigma \|_2.
\end{equation*}
An analogous bound of the form $\| \rho - \sigma \|_1 \leq 2 \sqrt{\mathrm{rank}(\sigma)}\| \rho - \sigma \|_2$ readily follows from exchanging the roles of $\rho$ and $\sigma$. Combining both implies
\begin{equation*}
   \| \rho - \sigma \|_1 \leq 2 \sqrt{\min \left\{ \mathrm{rank}(\rho),\mathrm{rank}(\sigma) \right\}} \| \rho - \sigma \|_2. 
\end{equation*}
\end{proof}
We note that a similar inequality was recently proved in
~\cite{Coles_2019}.
The above claim can be extended to effective rank. To this end, note that
\begin{equation*}
\tau (\rho,r) = \|\rho-\rho_{r} \|_1 \quad \text{where} \quad \rho_r = P_r \rho P_r (1-\tau (r,\rho))
\end{equation*}
and $P_r$ is the projection onto the range of the $r$ largest eigenvectors of $\rho$. We then have:

\begin{lem}\label{lem:conversionrank}
Let $\rho,\sigma$ be quantum states. Then for all $1\leq r\leq D-1$:
\begin{align*}
    \|\rho-\sigma\|_1\leq 2\sqrt{r}\|\rho-\sigma\|_2+2\min\{\tau(r,\rho),\tau(r,\sigma)\}.
\end{align*}
\end{lem}
\begin{proof}
Let $\tilde{\rho}_r=P_r \rho P_r$ be the best rank $r$ approximation of $\rho$ with respect to trace norm.
Decompose $\rho$ as $\rho=\tilde{\rho}_r+\tilde{\rho}_c$, with $\tilde{\rho}_c=(\mathbb{I}-P_r)\rho(\mathbb{I}-P_r)$ and apply a triangle inequality to conclude
\begin{align}\label{equ:distdecompositiontrace}
    \|\rho-\sigma\|_{1}\leq \|\tilde{\rho}_r-\sigma\|_1+\|\tilde{\rho}_c\|_{1}=\|\tilde{\rho}_r-\sigma\|_1+\tau(\rho,r).
\end{align}
Now, let $P_+$ be the orthoprojector onto the positive range of $\rho_r-\sigma$ and denote its orthocomplement by $P_- = \mathbb{I}-P_+$.
Then,
\begin{align}\label{equ:decompositiontrace}
    \|\tilde{\rho}_r-\sigma\|_1=\tr{P_+\lb \tilde{\rho}_r-\sigma\rb}-\tr{P_-\lb \tilde{\rho}_r-\sigma\rb}.
\end{align}
and the following similar identity is also true:
\begin{align*}
    \tau(\rho,r)=\tr{\sigma-\tilde{\rho}_r}=-\tr{P_+\lb \tilde{\rho}_r-\sigma\rb}-\tr{P_-\lb \tilde{\rho}_r-\sigma\rb}.
\end{align*}
Combining both yields
\begin{align}\label{equ:decompositionnegativepart}
    -\tr{P_-\lb \tilde{\rho}_r-\sigma\rb}=\tau(\rho,r)+\tr{P_+\lb \tilde{\rho}_r-\sigma\rb}.
\end{align}
Inserting Eq.~\eqref{equ:decompositionnegativepart} into~\eqref{equ:decompositiontrace} we conclude that
\begin{align*}
    \|\tilde{\rho}_r-\sigma\|_1=2\tr{P_+\lb \tilde{\rho}_r-\sigma\rb}+\tau(\rho,r)\leq 2\tr{P_+\lb \rho-\sigma\rb}+\tau(\rho,r).
\end{align*}
Finally, note that $P_+$ has rank at most $r$ by construction (it is the projector onto the positive range of $\tilde{\rho}_r - \sigma$ and $\tilde{\rho}_r$ has rank $r$) and therefore obeys $\|P_+ \|_2 \leq \sqrt{r}$. The Cauchy-Schwarz inequality thus asserts
\begin{align*}
  \|\tilde{\rho}_r-\sigma\|_1\leq  2\sqrt{r}\|\rho-\sigma\|_2+\tau(\rho,r).
\end{align*}
and the claim -- with $\tau (\rho,r)$ -- follows from combining this bound with Eq.~\eqref{equ:distdecompositiontrace}.
Exchanging the roles of $\rho$ and $\sigma$ provides a similar bound that features $2 \tau (\sigma,r)$ instead. Taking the minimum of both bounds establishes the claim.
\end{proof}

\begin{cor}\label{cor:effectiverank}
Let $\rho,\sigma$ be quantum states and $1\geq\varepsilon>0$ be given. Then
\begin{align*}
    \|\rho-\sigma\|_1\leq 2
r_{\mathrm{eff},\alpha}(\rho)^{\frac{1}{2}}\varepsilon^{-\frac{\alpha}{2(1-\alpha)}}\|\rho-\sigma\|_2+2\varepsilon\lb 1-\alpha\rb^\frac{1}{\alpha}.
\end{align*}

\end{cor}
\begin{proof}
From the tail decay estimate in Eq.~\eqref{equ:taildecay} we obtain:
\begin{align*}
    \tau\lb r_{\mathrm{eff},\alpha}(\rho)\varepsilon^{-\frac{\alpha}{1-\alpha}},\rho\rb\leq \varepsilon\lb 1-\alpha\rb^\frac{1}{\alpha}.
\end{align*}
The claim then follows from combining this estimate on the decay of $\tau(r,\rho)$ with Lemma~\ref{lem:conversionrank}.
\end{proof}
Thus, we see that a bound on the $\alpha$-R\'enyi of the target state $\rho$ allows us to estimate how well the Frobenius norm approximates the trace norm. Moreover, we recover the bound based on the rank in the limit $\alpha\to 0$.
Let us now restate Thm~\ref{thm:main-detail} incorporating the effective rank.
\begin{thm}[Re-statement of Theorem~\ref{thm:main-detail} with effective rank] \label{thm:main-stable}
Suppose that we wish to reconstruct a $D$-dimensional target state $\rho$ with effective rank $r_{\text{eff},\alpha}(\rho)$ up to accuracy $\cO(\epsilon)$ in trace distance with probability at least $1-\delta$. 
Then, Hamiltonian Updates -- Algorithm~\ref{alg:HUtomo_last} -- based on any basis measurement primitive with parameters $\theta_{\mathcal{E}}(\rho),\tau_{\mathcal{E}}(\rho)>0$ achieves this goal, provided that we make the following parameter choices:
\begin{align*}
\varepsilon =& \theta_{\mathcal{E}}(\rho)r_{\mathrm{eff},\alpha}(\rho)^{-\frac{1}{2}}\epsilon^{1+\frac{\alpha}{2(1-\alpha)}}& \text{(accuracy within the algorithm)},\\
\eta =& \varepsilon/8 = \theta_{\mathcal{E}}(\rho)r_{\mathrm{eff},\alpha}(\rho)^{-\frac{1}{2}}\epsilon^{1+\frac{\alpha}{2(1-\alpha)}}/8 & \text{(step size)}, \\
T=& \lceil 32\log(D)/\varepsilon^2 \rceil = \lceil 32 \log (D) r_{\mathrm{eff},\alpha}(\rho)\epsilon^{-2-\frac{\alpha}{\alpha-1}}/\theta_{\mathcal{E}}^2(\rho) \rceil & \text{(maximum number of iterations)}, \\
L=& \lceil \log(T) \log (1/\delta)/\tau_{\mathcal{E}}(\rho) \rceil 
& \text{(size of the control loops)}.
\end{align*}
This corresponds to at most
\begin{equation*}
M = TL = \log (1/\delta) T \log (T)/\tau_{\mathcal{E}}(\rho) = \tilde{\mathcal{O}}(r_{\mathrm{eff},\alpha}(\rho)\epsilon^{-2-\frac{\alpha}{\alpha-1}}/(  \tau_{\mathcal{E}}(\rho) \theta_{\mathcal{E}}(\rho)^2 ) \big)
\end{equation*}
different measurement settings.
\end{thm}

\section{Fast matrix vector multiplication for approximate unitary designs and Cliffords}\label{sec:fastvecmulti}
The computational speedups obtained by our algorithm relied on the fact that is possible to perform vector matrix multiplication in time $\tcO(D)$ for the unitaries used in the algorithm, what we called the FMVM property. Let us now show that indeed, all the measurement setups considered in this work have the aforementioned property.

Let us start with Cliffords in $D=2^n$ and a brief review of their properties. It is possible to specify a Clifford gate by a list of parameters $(\alpha,\beta,\gamma,\delta,p,s)$, where the first four parameters are $n\times n$ matrices with bits and $p,s$ are vectors with $n$ bits~\cite{Koenig_2014}.
We then have that a $C\in\mathrm{Cl}(D)$ specified by these parameters acts as:
\begin{align*}
    CX_jC^\dagger=(-1)^{p_j}\prod\limits_{i=1}^kX_i^{\alpha_{ij}}Z_i^{\beta_{ji}},\quad CZ_jC^\dagger=(-1)^{s_j}\prod\limits_{i=1}^kX_i^{\gamma_{ji}}Z_i^{\delta_{ji}},
\end{align*}
where $X_j$ and $Z_j$ the local Pauli operators.
Given these parameters, it is possible to find a circuit with $\cO(n^2)$ gates only consisting of Hadamard, CNOT and P gates in $\cO(n^2/\log(n))$ time~\cite{Aaronson_2004}. Moreover, the authors of~\cite{Koenig_2014} give a protocol to sample from the Clifford group efficiently in time $\cO(n^3)$ and whose output is given in terms of the aforementioned parameters. Thus, we conclude that for the Clifford group it is possible to sample, store a classical description and find a decomposition into $\cO(n^2)$ simple local gates in $\textrm{poly}(n)$ time. With this in  mind, we have:
\begin{lem}\label{lem:clifffast}
Let $C\in\mathrm{Cl}(D)$ be a random element of the Clifford group. Then we can compute $Cx$ for $x\in\C^D$ in time $\tcO(D)$.
\end{lem}
\begin{proof}
As remarked above, we can assume that the Clifford gate is presented as a sequence of $\cO(D^2)$ gates acting on at most $2$ qubits consisting of local Hadamard, CNOT and P gates. Now note that these local gates tensored with identity gates have at most $4$-sparse columns, as tensoring with the identity preserves the number of nonzero entries per column. Moreover, it is also possible to determine which entries are nonzero in $\cO(n)$ time.
Multiplying a vector with a $4$-sparse matrix with knowledge of the nonzero entries can be done in time $\cO(D)$. Thus, multiplying by each gate takes time $\cO(D)$. We conclude that we can multiply a vector by the sequence of gates in time $\tcO(D)$.
\end{proof}
Arguing in the same way as before, and noting that random local quantum circuits of polynomial depth give rise to approximate designs~\cite{Brand_o_2016} we also have that:
\begin{fact}
    Let $\mu$ be a $(4,\epsilon n^{-3})$ approximate unitary design given by a local quantum circuit on $n$ qudits of local dimension $d$ consisting of $\textrm{poly}(n)$ two qudit gates. Then for $U$ drawn from $\mu$ we can compute $Cx$ for $x\in\C^n$ in time $\tcO(Dd^2)$.
\end{fact}
\noindent
In a nutshell, we see that it is possible to store local circuits and random Cliffords using $\tcO(1)$ classical memory and perform matrix vector multiplication in time $\tcO(D)$, establishing the FMVM property for these relevant classes of measurement ensembles.

\end{document}